\documentclass[reprint,aps,pra,showkeys,floatfix,nofootinbib,showpacs,superscriptaddress]{revtex4-1} 

\usepackage{amsmath}
\usepackage{amsfonts}
\usepackage{amssymb}
\usepackage{graphicx}
\usepackage{color}
\usepackage[caption=false]{subfig}
\usepackage{hyperref}

\newtheorem{theorem}{Theorem}
\newtheorem{lemma}{Lemma}
\newtheorem{definition}{Definition}

\newenvironment{proof}[1][Proof]{\textbf{#1.} }{\ \rule{0.5em}{0.5em}}

\DeclareMathOperator{\Tr}{Tr}

\newcommand{\ket}[1]{ | \, #1  \rangle}
\newcommand{\bra}[1]{ \langle #1 \,  |}
\newcommand{\proj}[1]{\ket{#1}\bra{#1}}

\begin{document}

\title{Properties of dimension witnesses and their semi-definite programming relaxations}

\author{Piotr Mironowicz}
\email{piotr.mironowicz@gmail.com}
\affiliation{Department of Algorithms and System Modelling, Faculty of Electronics, Telecommunications and Informatics, Gda\'{n}sk University of Technology, Gda\'{n}sk 80-233, Poland}
\affiliation{National Quantum Information Centre in Gda\'{n}sk, Sopot 81-824, Poland}

\author{Hong-Wei Li}
\email{lihw@mail.ustc.edu.cn}
 \affiliation{Key Laboratory of Quantum Information,University of Science and Technology of China, Hefei, 230026, China}
 \affiliation{Zhengzhou Information Science and Technology Institute, Zhengzhou, 450004, China}

\author{Marcin Paw\l{}owski}
\email{dokmpa@univ.gda.pl}
\affiliation{Institute of Theoretical Physics and Astrophysics, University of Gda\'{n}sk, 80-952 Gda\'{n}sk, Poland}

\date{\today{}}

\begin{abstract}
In this paper we develop a method for investigating semi-device-independent randomness expansion protocols that was introduced in [Li {\it et al.} Phys. Rev. A \textbf{87}, 020302(R) (2013)]. This method allows to lower-bound, with semi-definite programming, the randomness obtained from random number generators based on dimension witnesses. We also investigate the robustness of some randomness expanders using this method. We show the role of an assumption about the trace of the measurement operators and a way to avoid it. The method is also generalized to systems of arbitrary dimension, and for a more general form of dimension witnesses, than it the previous paper. Finally, we introduce a procedure of dimension witness reduction, which can be used to obtain from an existing witness a new one with higher amount of certifiable randomness. The presented methods finds an application for experiments [Ahrens {\it et al.} Phys. Rev. Lett. {\bf 112}, 140401 (2014)].
\end{abstract}

\maketitle

\section{Introduction}

Nowadays information is one of the most important resources. We can defeat the enemy in a war just manipulating his data. If we can guess the mechanism of the generation of (pseudo-)random numbers used by a casino, then we can efficiently cheat in gambling\cite{Mitnick}. However, most of the so-called random number generators has a deterministic algorithm inside. It is very difficult to develop a reliable pseudo-random number generation (PRNG) method. Although there are tests\cite{NIST80022} that allow to check whether a sequence of numbers conforms to a particular probability distribution, we can never be sure its security without the knowledge how the sequence was generated.

One of the measures of randomness is so-called min-entropy\cite{OperMinEn}. In particular, in the context of authentication, min-entropy is the probability of guessing the easiest key in a given distribution of keys\cite{NIST800632}. If we know the pseudo-random generating algorithm and the initial seed (or some sequence of generated numbers), then the randomness of such a source is equal to zero. All classical PRNGs have this drawback.

On the other hand, quantum physics confuses philosophers with randomness on its deepest level. This randomness is unavoidable. We know that if certain observables (\emph{Bell operators}, which are linear functions of observed probabilities occurring in the experiment) attain certain thresholds, then the process must be intrinsically random, or we would have to abandon some ideas that are fundamental to all physical theories. Thus, values of these observables guarantee that the results of performed measurements are indeed random, no matter how does the measuring apparatus work.

This way the idea of the quantum randomness certification emerged\cite{RNGCBT}. If we want to be sure that the device we are using does really produce random numbers, we perform \textit{Bell experiment}, which is a kind of self testing. Such an experiment involves at least two separated parties that perform subsequent measurements with different settings without any communication between them. After series of such measurements, the collected data is used to estimate the joint probabilities of the outcomes conditioned on the settings used. The most prominent example of Bell operator is so-called Clauser-Horne-Shimony-Holt (CHSH)\cite{CHSH}.

Because such self testing works independently of the internal workings of the device used (in particular, the exact form of the performed measurements is not important), if the Bell inequality attains some value, we are sure that the generated results are indeed random, even if the device has been construed by a malevolent party. The amount of the obtained secure randomness is precisely quantified by means of min-entropy\cite{RNGCBT,ColPHD,CK11}. This approach, in which we do not trust the vendor of our devices and draw conclusions only from the observed results, is called \textit{device-independent}\cite{DI}, referred further as \textit{DI}.

Still, Bell experiments are very difficult to do. They require a high degree of precision and extremely high detection efficiencies. So far loophole-free Bell experiment has not been successfully performed. But when we allow to send a state from one part of the device to another, then we do not have any non-locality, which is crucial for that way of certification. It was shown that, if we can bound the dimension of the communicated system, we still may use this \textit{prepare and measure scheme} to certify the randomness\cite{DW1}. Since we have to know something about the construction of the device, this approach is called \textit{semi-device-independent}\cite{SDI,SDI_effects} (denoted hereafter \textit{SDI}). This offers a good compromise between security and experimental feasibility.

Currently all commercial quantum random generators are based on the prepare and measure scheme, \textit{e.g.} the \textit{id~Quantique}'s\cite{IDQ} device \textit{Quantis}, or the \textit{qStream} by \textit{Quintessence Labs}\cite{QLabs}. These devices do not perform any self testing, so we are forced to trust their vendors. For this reason, methods for certifying randomness in the prepare and measure scheme with the semi-device-independent approach should be investigated. In this framework analogs of Bell inequalities, called \textit{dimension witnesses} \cite{DW1,DW4,DW2,DW3,DW5}, are used.

Before we proceed we should stress that what we call random number generation is in fact randomness \textit{expansion}, the process that starts with some amount of initial randomness and uses it to obtain more of it. The presented self testing procedure of the device also requires some amount of randomness (in order to choose the measurement settings in rounds of testing experiments). Strictly speaking, all quantum random number generators that use Bell inequalities or dimension witnesses to certify the randomness are randomness expanders.

After generation of a string of bits with a certain amount of min-entropy, it is possible to \textit{extract} its randomness what means using a certain algorithm to produce a shorter string with a larger min-entropy per bit\cite{Trevisan01,DPVR09,TRSS10}.

In our previous paper \cite{HWL13} we have investigated the relation between random number expansion protocols based on correlations occurring in the scenario where two parties share an entangled state, and on protocols relying on the prepare and measure scheme. In this paper we develop these ideas.

The organization of this paper is as follows. In the section \ref{sec:Motivation} we present scenario in which we are working. Next, in sections \ref{sec:BI} and \ref{sec:DW}, we give basic information about Bell inequalities and dimension witnesses. Then, in the section \ref{sec:BItoDW}, we recapitulate a heuristic method of obtaining a dimension witness from a Bell inequality. This method was introduced in \cite{HWL13}. In sections \ref{sec:DWtoBI} and \ref{sec:binaryDW} we precisely state the conditions when the randomness certified by the violation of a Bell inequality lower-bounds the randomness certified by a certain value of dimension witness in the semi-device-independent scenario. In the section \ref{sec:symDW} we investigate the properties of a certain class of dimension witnesses and introduce a procedure of dimension witness reduction, which can be used to obtain from an existing witness a new one with higher amount of certifiable randomness.  In the section \ref{sec:explicitExamples} we give examples of application of the presented methods.

The aims of this paper are as follows. We clarify the methods from our previous paper \cite{HWL13} and give a tighter lower bound on randomness. Using these methods we obtain better dimension witnesses, in particular the one based on the Braunstein-Caves Bell inequality \cite{BC88}. We also extend the applicability of the methods from \cite{HWL13} to arbitrary dimensions.

\section{Motivation of this paper}
\label{sec:Motivation}

Suppose we are a developer of a random number generating device. Since consumers do not trust us, we are interested in finding a way of certification for our device. Common method for the certification of quantum random number generators that are based on measurements on entangled particles is to estimate the value of a certain Bell inequality that is attained in this device. Still, it is too difficult to observe a loophole-free violation of Bell inequality. Thus we prefer prepare and measure protocols.

Both for prepare and measure protocols in the semi-device-independent approach, and for correlation protocols in the device-independent scheme, we would like to define a value that measures how reliable is its particular realization. As this value we take the expectation value of the relevant dimension witness or Bell inequality, respectively, attained in the relevant protocol. This value is called a \textit{security parameter}.

It is possible to consider several relations. One may ask whether, having a protocol of one type, we can relate it to some protocol of another type, in such a way that for the same value of their security parameters the min-entropy certified in one of them, is upper or lower bounded by min-entropy certified by the other one. One may start with a protocol based on a Bell inequality and construct out of it a prepare and measure protocol certifying a reasonable amount of min-entropy. This is useful since there are many randomness expansion protocols based on Bell inequalities\cite{RNGCBT, CK11} and it is easy to obtain new ones \cite{MP13}.

Another situation is when we begin with some SDI protocol and want to lower bound the certified randomness using efficient numerical methods from \cite{NPA07, NPA08}, that works in the device-independent approach. We present a way to obtain a new Bell inequality with the property that the DI protocol using it certifies at most as much randomness as the SDI protocol.

As mentioned above, SDI protocols are much easier to implement than the protocols based on entanglement. For this reason it is useful to have a method that allows to develop devices of the first kind with the help of the well established knowledge about the devices of the second type.

\section{Bell inequalities}
\label{sec:BI}

We define for a DI protocol:
\begin{definition}
    \label{probDI}
    Let $A$, $B$, $X$, and $Y$ be sets.
		
		Probability distribution in DI scheme is a conditional probability distribution $\mathbb{P}(A,B|X,Y)$ such that
    \begin{equation}
        \forall_{a \in A} \forall_{b \in B} \forall_{x \in X} \forall_{y \in Y} P(a,b|x,y) = \Tr \left(\rho M^a_x M^b_y\right), \nonumber
    \end{equation}
    where $\{\{M^a_x\}_{a \in A}\}_{x \in X}$ and $\{\{M^b_y\}_{b \in B}\}_{y \in Y}$ are sets of POVMs on a Hilbert space $\mathbb{H}$, and $\rho$ is a density matrix on $\mathbb{H}$, and
    \begin{equation}
        \forall_{a \in A} \forall_{b \in B} \forall_{x \in X} \forall_{y \in Y} [M^a_x, M^b_y] = 0 \text{ if } x \neq y.
    \end{equation}
    We denote this probability by
    \begin{equation}
        \mathbb{P}[\rho, \{\{M^a_x\}_{a \in A}\}_{x \in X}, \{\{M^b_y\}_{b \in B}\}_{y \in Y}]. \nonumber
    \end{equation}

    If $A = B = \{0,1\}$, then $\mathbb{P}(A,B|X,Y)$ is called binary.

    The set of all DI probability distributions for given $A$, $B$, $X$ and $Y$ is denoted by $\mathcal{P}(A,B|X,Y)$.
\end{definition}

Let us take two sets, $X$ and $Y$, that label the measurement settings of Alice and Bob in DI scheme, and two sets, $A$ and $B$, that label their respective outcomes.

A Bell inequality is a linear function defined, in particular, for probability distributions $\mathcal{P}(A,B|X,Y)$. It is of the form
\begin{equation}
    \label{BI}
    \begin{aligned}
        I & (A,B,X,Y,\{\alpha_{a,b,x,y}\},C_I)[\mathbb{P}(A,B|X,Y)] \equiv \\
        &   = \sum_{a \in A} \sum_{b \in B} \sum_{x \in X} \sum_{y \in Y} \alpha_{a,b,x,y} P(a,b|x,y) + C_I,
    \end{aligned}
\end{equation}
where $\alpha_{a, b, x, y}, C_I \in \mathbb{R}$. We omit $\mathbb{P}$ if it is obvious which probability distribution is considered.

The constant term $C_I$ in a Bell inequality does not change its properties. Still, we retain this general form, both for Bell inequalities, and dimension witnesses in the next section. In the following sections this allows to keep the same maximal expected value when performing a transformation leading from one expression to another.

A particular form of Bell inequality is the following correlation form
\begin{equation}
    \label{BIhat}
    \begin{aligned}
        \hat{I} & (X, Y, \{\alpha_{x,y}\}, \hat{C}_I)[\mathbb{P}(\{0,1\},\{0,1\}|X,Y)] \equiv \\
        & = \sum_{x \in X} \sum_{y \in Y} \hat{\alpha}_{x,y} C(x,y) + \hat{C}_I,
    \end{aligned}
\end{equation}
with $\hat{\alpha}_{x,y},\hat{C}_I \in \mathbb{R}$, and
\begin{equation}
	\begin{aligned}
    C(x,y) &= P(0,0|x,y) - P(0,1|x,y) \\
		& -P(1,0|x,y) + P(1,1|x,y).
	\end{aligned} \nonumber
\end{equation}
Obviously, the form (\ref{BI}) conforms the form (\ref{BIhat}) if, and only if $\alpha_{0,0,x,y} = \alpha_{1,1,x,y} = -\alpha_{0,1,x,y} = -\alpha_{1,0,x,y} = \hat{\alpha}_{x,y}$, and $\mathbb{P}(A,B|X,Y)$ is binary.

For given $A$, $B$, $X$, $Y$, $x_0 \in X$, $y_0 \in Y$, Bell inequality $I$ and $s \in \mathbb{R}$ we define the following terms:
\begin{subequations}
    \begin{equation}
        P_{guess}(\mathbb{P}(A,B|X,Y),x_0,y_0) \equiv \max_{a \in A, b \in B} P(a,b|x_0,y_0) \nonumber
    \end{equation}
    \begin{equation}
        \nonumber
        \begin{aligned}
            H_{\infty} & (\mathbb{P}(A,B|X,Y),x_0,y_0) \equiv \\
            & -\log_2 \left( P_{guess}(\mathbb{P}(A,B|X,Y),x_0,y_0) \right)
        \end{aligned}
    \end{equation}
    \begin{equation}
        \nonumber
        \begin{aligned}
            H_{\infty}^{cert} & (I,x_0,y_0,s) \equiv \\
            & \min_{\mathbb{P}(A,B|X,Y) \in \mathcal{P}(A,B|X,Y)} H_{\infty}(\mathbb{P}(A,B|X,Y),x_0,y_0), \\
            & \text{subject to } I[\mathbb{P}(A,B|X,Y)] \geq s
        \end{aligned}
    \end{equation}
\end{subequations}
The expression $H_{\infty}(\mathbb{P}(A,B|X,Y),x_0,y_0)$ is called min-entropy, and $H_{\infty}^{cert}(I,x_0,y_0,s)$ is the min-entropy certified by the value $s$ of $I$.

\section{Dimension witnesses}
\label{sec:DW}

For a SDI scheme, we have the following definition of the allowed
probability distribution
\begin{definition}
    \label{probSDI}
    Let $\bar{B}$, $\bar{X}$, and $\bar{Y}$ be sets, and $\mathbb{H}$ be a Hilbert space of a finite dimension $d$.
		
		A probability distribution in SDI scheme is a conditional probability distribution $\mathbb{P}_d (\bar{B}|\bar{X},\bar{Y})$ such that for $b \in \bar{B}$, $x \in \bar{X}$ and $y \in \bar{Y}$ we have $P(b|x,y) = \Tr \left(\rho_{x} M^b_y\right)$, where $\{\rho_{x}\}_{x \in \bar{X}}$ is a set of density matrices on $\mathbb{H}$, and $\{M^b_y\}_{b \in \bar{B}}$ are POVMs on $\mathbb{H}$ for all $y \in \bar{Y}$.
		
		We say that $\mathbb{P}_d$ is realized by sets $\{\rho_x\}_{x \in \bar{X}}$ and $\{\{M^b_y\}_{b \in \bar{B}}\}_{y \in \bar{Y}}$, and denote it
    \begin{equation}
        \mathbb{P}_d \left[ \{\rho_x\}_{x \in \bar{X}}, \{\{M^b_y\}_{b \in \bar{B}}\}_{y \in \bar{Y}} \right]. \nonumber
    \end{equation}

    If $\bar{B} = \{0,1\}$, then $\mathbb{P}_d(\bar{B}|\bar{X},\bar{Y})$ is called a binary probability distribution.

    The set of all SDI probability distributions for given $d$, $\bar{B}$, $\bar{X}$ and $\bar{Y}$ is denoted by $\mathcal{P}_d (\bar{B}|\bar{X},\bar{Y})$.

    The set of all SDI probability distributions with restrictions that $d = 2$, $\bar{B} = \{0,1\}$ and $\forall_{b \in \{0,1\}} \forall_{y \in \bar{Y}} \Tr M^b_y = 1$ is denoted by $\mathcal{P}^{(P)}(\bar{X},\bar{Y})$.
\end{definition}

Let $\bar{X}$ and $\bar{Y}$ be sets labeling the settings of Alice and Bob, in the SDI scheme, and let $\bar{B}$ be a set of the outcomes that Bob can obtain.

Dimension witnesses are linear functions of probability distributions of the form
\begin{equation}
    \label{DW}
    \begin{aligned}
        W & (\bar{B}, \bar{X}, \bar{Y}, \{\beta_{b,x,y}\}, C_W)[\mathbb{P}_d(\bar{B}|\bar{X},\bar{Y})] \equiv \\
        & = \sum_{b \in \bar{B}} \sum_{x \in \bar{X}} \sum_{y \in \bar{Y}} \beta_{b,x,y} P(b|x,y) + C_W,
    \end{aligned}
\end{equation}
where $\beta_{b,x,y}, C_W \in \mathbb{R}$, and $d \geq 2$.

If $\bar{B} = \{0,1\}$, then the dimension witness is called binary. If $\forall_{b \in \bar{B}} \forall_{y \in \bar{Y}} \sum_{x \in \bar{X}} \beta_{b,x,y} = 0$, then the dimension witness is called zero-summing.

For given $\bar{B}$, $\bar{X}$, $\bar{Y}$, $x_0 \in \bar{X}$, $y_0 \in \bar{Y}$, dimension witness $W$, $s \in \mathbb{R}$ and $d \geq 2$ we define the following terms:
\begin{subequations}
    \begin{equation}
				\label{Pguess}
        P_{guess}(\mathbb{P}_d(\bar{B}|\bar{X},\bar{Y}),x_0,y_0) \equiv \max_{b \in \bar{B}} P(b|x_0,y_0)
    \end{equation}
    \begin{equation}
        H_{\infty}(\mathbb{P}_d(\bar{B}|\bar{X},\bar{Y}),x_0,y_0) \equiv -\log_2 \left( P_{guess}(\mathbb{P}_d(\bar{B}|\bar{X},\bar{Y}),x_0,y_0) \right) \nonumber
    \end{equation}
    \begin{equation}
        \label{PguessCert}
        \begin{aligned}
            P_{guess}^{cert} & (W,x_0,y_0,s,d) \equiv \max_{\mathbb{P}_d(\bar{B}|\bar{X},\bar{Y}) \in \mathcal{P}_d (\bar{B}|\bar{X},\bar{Y})} \max_{b \in \bar{B}} P(b|x_0,y_0), \\
            & \text{subject to } W[\mathbb{P}_d(\bar{B}|\bar{X},\bar{Y})] \geq s
        \end{aligned}
    \end{equation}
    \begin{equation}
        H_{\infty}^{cert}(W,x_0,y_0,s,d) \equiv -\log_2 \left( P_{guess}^{cert}(W,x_0,y_0,s,d) \right) \nonumber
    \end{equation}
    \begin{equation}
        \label{PguessCertP}
        \begin{aligned}
            P_{guess}^{cert(P)} & (W,x_0,y_0,s) \equiv \\
            & \max_{\mathbb{P}_d(\bar{B}|\bar{X},\bar{Y}) \in \mathcal{P}^{(P)}(\bar{X},\bar{Y},s)} \max_{b \in \bar{B}} P(b|x_0,y_0).
        \end{aligned}
    \end{equation}
\end{subequations}
The expression $H_{\infty}(\mathbb{P}_d(\bar{B}|\bar{X},\bar{Y}),x_0,y_0)$ is called min-entropy, and $H_{\infty}^{cert}(W,x_0,y_0,s)$ is the min-entropy certified by the value $s$ of $W$ (for the dimension $d$).

The following lemma summarizes some properties of dimension witnesses.
\begin{lemma}
    \label{dwLemma}
    Let $\mathbb{H}$ be a Hilbert space of a dimension $2$, and let $W$ be a binary dimension witness defined by certain $\bar{X}$, $\bar{Y}$, $\{\beta_{b,x,y}\}$ and $C_W$.
		
		Let $\{\rho_x\}_{x \in \bar{X}} \equiv \mathcal{S}$ be a set of states on $\mathbb{H}$, and $\{\{M^0_y,M^1_y\}\}_{y \in Y} \equiv \mathcal{M}$ be a set of binary POVM on $\mathbb{H}$. Let $s \equiv W[\mathbb{P}_2 (\mathcal{S},\mathcal{M})]$.
		
		Then, the following implications hold:
    \begin{enumerate}
        \item If $\forall_{y \in \bar{Y}} \sum_x \beta_{0,x,y} = \sum_x \beta_{1,x,y}$, then there exists a set of binary POVM on $\mathbb{H}$, $\tilde{\mathcal{M}} \equiv \{\{\tilde{M}^0_y,\tilde{M}^1_y\}\}_{y}$, such that $\forall_{y,b} \Tr \tilde{M}^b_y = 1$, and $W[\mathbb{P}_2 (\mathcal{S},\tilde{\mathcal{M}})] = s$.
        \item If $\sum_{b,x,y} \beta_{b,x,y} = 0$, and $\forall_{y,b} \Tr M^b_y = 1$, then for  $\tilde{\mathcal{S}} = \{\openone - \rho_x\}_{x \in \bar{X}}$, which is a set of states on $\mathbb{H}$, $W[\mathbb{P}_2 (\tilde{\mathcal{S}},\mathcal{M})] = -s$.
        \item If $\forall_{y,b} \Tr M^b_y = 1$, then there exist a set of projective measurements, $\tilde{\mathcal{M}} \equiv \{\{\Pi^0_y,\Pi^1_y\}\}_{y \in \bar{Y}}$ with $\forall_{b \in \bar{B}, y \in \bar{Y}} \Tr \left( \Pi^b_y \right) = 1$, such that $W[\mathbb{P}_2 (\mathcal{S},\tilde{\mathcal{M}})] \geq s$.
    \end{enumerate}
\end{lemma}
\begin{proof}
		\begin{enumerate}
				\item Let us take $y \in \bar{Y}$. Let $c_y = \frac{1}{2} \left(1 - \Tr(M^0_y)\right)$, $\tilde{M}^0_y = M^0_y + c_y \openone$, and $\tilde{M}^1_y = M^1_y - c_y \openone$. Obviously
				\begin{equation}
						\tilde{M}^0_y + \tilde{M}^1_y = M^0_y + M^1_y = \openone. \nonumber
				\end{equation}

				Now, we prove that $\forall_{y,b} \tilde{M}^b_y \succeq 0$. There exist an orthonormal basis $\{\ket{0_y},\ket{1_y}\}$ in that
				\begin{equation}
						M^0_y = v_0 \proj{0_y} + v_1 \proj{1_y}, \nonumber
				\end{equation}
				and
				\begin{equation}
						M^1_y = (1-v_0) \proj{0_y} + (1-v_1) \proj{1_y}, \nonumber
				\end{equation}
				where $v_0,v_1 \in [0,1]$. We have $c_y=\frac{1}{2}(1-v_0-v_1)$, and $\openone = \proj{0_y}+\proj{1_y}$. Thus
				\begin{equation}
						\tilde{M}^0_y = \frac{1}{2}(1+v_0-v_1)\proj{0} + \frac{1}{2}(1-v_0+v_1)\proj{1}. \nonumber
				\end{equation}
				Since $1+v_0-v_1\geq 0$ and $1-v_0+v_1 \geq 0$, we have $\tilde{M}^0_y \succeq 0$, and $\Tr \tilde{M}^0_y = 1$. Similarly, we check that $\tilde{M}^1_y \succeq 0$ and $\Tr \tilde{M}^1_y = 1$.

				Repeating this construction for all $y \in \bar{Y}$, we obtain a set of POVM, $\tilde{\mathcal{M}} \equiv \{\{\tilde{M}^0_y,\tilde{M}^1_y\}\}_{y \in \bar{Y}}$.

				We have
				\begin{equation}
						\Tr (\rho_x \tilde{M}^b_y) = P(b|x,y) + (-1)^b \cdot c_y, \nonumber
				\end{equation}
				and thus
				\begin{equation}
						\begin{aligned}
								W & [\mathbb{P}_d(\mathcal{S},\tilde{\mathcal{M}})] = \sum_{b,x,y} \beta_{b,x,y} \Tr (\rho_x \tilde{M}^b_y) \\
								& = s + \sum_y c_y \left( \sum_x \beta_{0,x,y} - \sum_x \beta_{1,x,y} \right) = s. \nonumber
						\end{aligned}
				\end{equation}

				\item We have
				\begin{equation}
						\begin{aligned}
								W[\mathbb{P}_2 & (\tilde{\mathcal{S}},\mathcal{M})] = \sum_{b,x,y} \beta_{b,x,y} \Tr \left( (\openone - \rho_x) M^b_y \right) \\
								& = \sum_{b,x,y} \beta_{b,x,y} (1 - P(b|x,y)) \\
								& = \sum_{b,x,y} \beta_{b,x,y} - \sum_{b,x,y} \beta_{b,x,y} P(b|x,y) = -s. \nonumber
						\end{aligned}
				\end{equation}

				\item For any $y \in \bar{Y}$ we have
				\begin{equation}
						M^0_y = \lambda_y \proj{0_y} + (1-\lambda_y) \proj{1_y}, \nonumber
				\end{equation}
				and
				\begin{equation}
				M^1_y = (1-\lambda_y) \proj{0_y} + \lambda_y \proj{1_y}, \nonumber
				\end{equation}
				for a certain basis $\{\ket{0_y}, \ket{1_y}\}$, $\lambda_y \in [0,1]$.

				Let us define $s_y \equiv \sum_{b,x} \beta_{b,x,y} P(b|x,y)$. Denote
				\begin{equation}
						\sum_x \left( \beta_{0,x,y} \Tr (\rho_x \proj{0_y}) + \beta_{1,x,y} \Tr (\rho_x \proj{1_y}) \right) \nonumber
				\end{equation}
				by $s_{y,0}$, and similarly
				\begin{equation}
						\sum_x \left( \beta_{0,x,y} \Tr (\rho_x \proj{1_y}) + \beta_{1,x,y} \Tr (\rho_x \proj{0_y}) \right) \nonumber
				\end{equation}
				by $s_{y,1}$.

				We have $s = \sum_y s_y$, and
				\begin{equation}
						s_y = \lambda_y s_{y,0} + (1-\lambda_y) s_{y,1}. \nonumber
				\end{equation}
				If $s_{y,0} \geq s_{y,1}$, then we take $\tilde{M}^0_y \equiv \proj{0_y}$ and $\tilde{M}^1_y \equiv \proj{1_y}$, otherwise we take $\tilde{M}^0_y \equiv \proj{1_y}$ and $\tilde{M}^1_y \equiv \proj{0_y}$. For $\tilde{\mathcal{M}} = \{\{\tilde{M}^0_y,\tilde{M}^0_y\}\}_{y \in \bar{Y}}$ it is easy to see that
				\begin{equation}
						W[\mathbb{P}_2 (\mathcal{S},\tilde{\mathcal{M}})] = \sum_y \max(s_{y,0},s_{y,1}) \geq \sum_y s_y = s. \nonumber
				\end{equation}
		\end{enumerate}
\end{proof}

The first statement in this lemma says that in the dimension $2$ the condition that all measurement operators have trace $1$ is not restrictive with regards to the set of values possible to attain. The second statement gives sufficient conditions under which an operation of negation of all states gives the same value of a dimension witness but with opposite sign. The third statement, which may be used to complement the first one, shows that under certain conditions it is not restrictive to use only projective measurements in case when the values possible to be attained are considered.

\section{A heuristic method for obtaining a dimension witness from a Bell inequality}
\label{sec:BItoDW}

Consider the following Bell experiment. Suppose we are given a Bell inequality of the form (\ref{BI}). Alice and Bob share an entangled state. Alice chooses a measurement setting $x \in X$, and obtains an outcome $a \in A$. For each setting $x$ and result $a$, we assign a conditional probability $P(a|x)$. Alice's measurement prepares some state at Bob's side. Next, Bob chooses a measurement setting $y \in Y$, and obtains an outcome $b \in B$. The probability that Bob gets $b$, knowing both the setting and the result of Alice, is $P(b|a,x,y)$.

We rewrite\footnote{We are using here the no-signaling principle.} the joint conditional probability of a given pair of results for a given pair of settings as $P(a,b|x,y)=P(b|a,x,y) \cdot P(a|x)$. Thus, defining $\bar{x} \equiv (a,x)$, the initial Bell inequality is transformed to the \textit{form} of a dimension witness (see the equation (\ref{DW})), with $\beta_{b,\bar{x},y} \equiv \beta_{b,(a,x),y} \equiv \alpha_{a,b,x,y} \cdot P(a|x)$. We have $\bar{B} = B$, $\bar{X} = A \times X$, and $\bar{Y} = Y$.

The fact that it is possible to transform a Bell inequality into the form of a dimension witness, leads us to some \textit{heuristic} method to achieve an SDI protocol that certifies a reasonable amount of randomness, once we have a DI protocol. We get the SDI protocol if, instead of measuring on Alice's side, she gets "the outcome" as a part of her input with the probability distribution $P(a|x)$. Thus, we obtain a pair $(a,x)$ that we use as an index of the state to be send. This way, the device on the side of Alice prepares one of $|\bar{X}| = |A| \cdot |X|$ states $\rho_{(a,x)}$. Bob still has $|\bar{Y}| = |Y|$ measurement settings.

\section{Lower-bounds for dimension witnesses via semi-definite programs}
\label{sec:DWtoBI}

In this section we construct a sequence of devices that shows that the randomness certified by an SDI protocol can be lower bounded by the randomness certified in a certain DI protocol minus $\log_2{d}$.

We consider a device $\textbf{D0}$ that we get from an untrusted vendor, and that consists of two black boxes. Its only parameter that we can verify (or trust), is the dimension of the message send from one part of it, to the another one. We assume, that the device cannot communicate with the world outside the laboratory. The black box on Alice's side has buttons with labels $x \in \bar{X}$ and emits one of the states of the dimension $d$ from the set of states $\{\rho_{x}\}_{x \in \bar{X}}$. The states are unknown to us, and are of arbitrary, possibly mixed, form. The black box on Bob's side has buttons with labels $y \in \bar{Y}$ and, after receiving the qubit from Alice's black box, it performs one of the measurements given by POVMs from the set$\{\{M^b_y\}_{b \in \bar{B}}\}_{y \in \bar{Y}}$. We do not know, how the measurements are performed. This description is semi-device~independent, since we know only the dimension $d$.

Suppose we are given a dimension witness W (of the form (\ref{DW})) that achieves in the experiments on the device $\textbf{D0}$ the expected value $W_{0}$. We denote the conditional probability of obtaining the outcome $b$ when the chosen settings are $x$ and $y$, by $P_{\textbf{D0}}(b | x, y)$.

The device $\textbf{D0}$ is not trusted, but it is possible to consider another device, $\textbf{D1}$, that consists of two parts, with buttons labeled by $x \in \bar{X}$ and $y \in \bar{Y}$ on the Alice's side and on the Bob's side, respectively. The parts are sharing a maximally entangled state of the dimension $d$. The part on the Alice's side performs some measurement, depending on the chosen input $x$. This measurement projects the Alice's part of the singlet on the state $\rho_{x}$ that is the same as the relevant state from the device $\textbf{D0}$. If the projection succeeded, which happens with the probability $\frac{1}{d}$, then the device returns $a = 0$ and changes the state on the Alice's side into the state $\rho_{x}$, otherwise it returns $a = 1$. Since the shared state is a singlet, this measurement prepares the same $d$-dimensional state on the Bob's side. Then he performs the same POVM $\{M^b_y\}_{b \in \bar{B}}$ as the device $\textbf{D0}$, and returns the outcome $b \in \bar{B}$.

The probability that Alice gets the outcome $a$ with the setting $x$, and simultaneously Bob gets the outcome $b$ with the setting $y$ is denoted by $P_{\textbf{D1}}(a,b|x,y)$. It is easy to see, that $P_{\textbf{D0}}(b|x,y) = d \cdot P_{\textbf{D1}}(0,b|x,y)$.

Now let us consider another device, $\textbf{D2}$. It has the same interface like $\textbf{D1}$, but the conditions on the internal working are relaxed, \textit{viz.} we do not assume anything about the performed measurements, and Alice's and Bob's parts are allowed to share any, possibly entangled, state $\rho$ of an arbitrary dimension. The probability of obtaining the outcomes $a$ and $b$ with given pair of settings $x$ and $y$ for Alice and Bob, respectively, are denoted by $P_{\textbf{D2}}(a,b|x,y)$. We apply a constraint $\forall_{x \in X} P_{\textbf{D2}}(0|x) = \frac{1}{d}$, where $P_{\textbf{D2}}(a|x)$ is the probability of getting the outcome $a$ by Alice with the setting $x$ with the device $\textbf{D2}$.

Obviously, all the conditional probability distributions that are possible to be obtained by the device $\textbf{D1}$ (and thus also by the device $\textbf{D0}$), are also possible to be obtained by this device. Note that this description is fully device-independent, and that there are semi-definite programs in the NPA hierarchy \cite{NPA07, NPA08} that efficiently approximate the probability distributions of the device $\textbf{D2}$.

Since the device $\textbf{D2}$ is a relaxed version of the initial device $\textbf{D0}$, if both of them have the same value of the relevant security parameters, then the certified amount of min-entropy generated by the device $\textbf{D2}$ gives a lower bound of the min-entropy certified to be generated by the device $\textbf{D0}$.

We recapitulate the above results in the following theorem
\begin{theorem}
    \label{DWtoBItheorem}
    Let $B = \bar{B}$, $X = \bar{X}$ and $Y = \bar{Y}$ be sets. Let us take $s \in \mathbb{R}$, $d \geq 2$, a Bell inequality $I$ of the form (\ref{BI}), and a dimension witness $W$ of the form (\ref{DW}), satisfying $\beta_{b,x,y} = d \cdot \alpha_{0,b,x,y}$.

    Let $\mathcal{P}_{d,SDI}(s)$ be a subset of $\mathcal{P}_d (\bar{B}|\bar{X},\bar{Y})$ with $d \geq 2$ (see the definition \ref{probSDI}) that satisfies $W = s$.

    Let $\mathcal{P}_{DI}(s)$ be a set of all probability distribution defined by $P(b|x,y) \equiv d \cdot P(0,b|x,y)$, where $\mathbb{P}(A,B|X,Y)$ is a device-independent probability distribution such that $I[\mathbb{P}(A,B|X,Y)] = s$, with $A = \{0,1\}$.

    Then $\mathcal{P}_{d,SDI}(s) \subseteq \mathcal{P}_{DI}(s)$.
\end{theorem}

This way we obtain a way to get a relation between Bell inequalities and dimension witnesses with the property that the amount of randomness certified by a Bell inequality lower-bounds the amount of randomness certified by the  relevant dimension witness. One of key features of the set $\mathbb{P}_{DI}$ is that it can be efficiently approximated using semi-definite programming with the NPA hierarchy.

From the definition of $\mathcal{P}_{DI}(B|X,Y)$, namely using $P(b|x,y) = d \cdot P(0,b|x,y)$, we get that the certified min-entropy of SDI protocol is lower-bounded by the one of the DI protocol minus $\log_2{d}$. Notable property of the method is that we obtain a bound for any dimension of the communicated system changing only a value of the linear bound. 

\begin{figure}[!htbp]
    \includegraphics[width=0.45\textwidth]{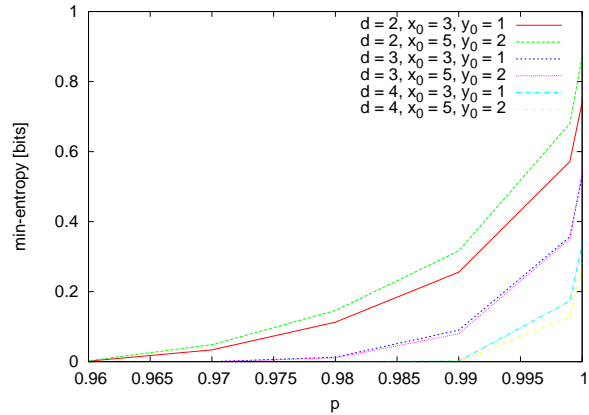}
    \caption{(Color on-line) Lower-bounds via SDP on the certified randomness for a dimension witness obtained from the CGLMP inequality (see the equation(\ref{CGLMPdw}) in section \ref{sec:explicitExamples}) using the methods from the section \ref{sec:DWtoBI} for different values of the dimension $d$.\label{fig:CGLMP}}
\end{figure}

The figure \ref{fig:CGLMP} shows an example of application of the theorem \ref{DWtoBItheorem}.

\section{Binary zero-summing dimension witnesses}
\label{sec:binaryDW}

In this section the properties of binary zero-summing dimension witnesses are investigated. Recall that a dimension witness of the form given by the equation (\ref{DW}) is called zero-summing if $\forall_{b \in \bar{B}} \forall_{y \in \bar{Y}} \sum_{x \in \bar{X}} \beta_{b,x,y} = 0$, and binary if $\bar{B} = \{0,1\}$. The reason to examine them is that it is possible to obtain a tighter semi-definite relaxations for this class of dimension witnesses.

Let us start with a binary zero-summing dimension witness $W = W(\{0,1\},\bar{X},\bar{Y},\{\beta_{b,x,y}\},C_W)$ that is used to certify the randomness generated by measuring the state $x_0 \in \bar{X}$ with the measurement setting $y_0 \in \bar{Y}$. Let $\{\rho_x\}_{x \in \bar{X}}$ and $\{\{M^0_y,M^1_y\}\}_{y \in \bar{Y}}$ be the states and measurements that maximize the guessing probability (see the equation (\ref{Pguess})) of the generated bits by the untrusted vendor.

First note that the value of the dimension witness does not change if, for arbitrary $y \in \bar{Y}$, the measurement is changed to $\{M^0_y + c \openone, M^1_y - c \openone\}$, where $c$ is such that the spectrum of the operators remains in the range $[0,1]$. Thus, since the potential adversary is interested in increasing the probability of a particular outcome of the measurement $y_0$ as much as possible, the form of these measurements that maximizes his guessing probability is the following:
\begin{equation}
    \label{optMeasY0}
    \begin{array}{l l}
        \begin{bmatrix}
            1 & 0 \\
            0 & 1 - \delta
        \end{bmatrix},
        & \quad
        \begin{bmatrix}
            0 & 0 \\
            0 & \delta
        \end{bmatrix}
    \end{array}
\end{equation}

By the lemma \ref{dwLemma}.1 and \ref{dwLemma}.3 it is not restrictive for the vendor  to use only projectors of trace $1$ for the measurements different than $y_0$.

The strategy of using a measurement of the form (\ref{optMeasY0}) for the setting $y_0$, and projectors of trace $1$ for all remaining measurements is equivalent to using the following mixed strategy. In $\delta$ cases, a projective measurement of trace $1$ is used for the measurement $y_0$ (we call this strategy:  P), and in $1 - \delta$ cases the outcome is deterministic - this is referred hereafter as a deterministic strategy, or simply: D. For the remaining measurements the same projective measurements of trace $1$ are used in both cases.

The guessing probability for the strategy D is $1$, and for the strategy P is $p$, thus the average guessing probability is
\begin{equation}
    \label{deltaAffine}
    (1 - \delta) + \delta \cdot p.
\end{equation}

In the case of a zero-summing dimension witness with the deterministic strategy, measurements with the setting $y_0$ give no contribution to the value of the witness. Thus the certification of the randomness with the dimension witness
\begin{equation}
    W = W(\{0,1\},\bar{X},\bar{Y},\{\beta_{b,x,y}\},C_W) \nonumber
\end{equation}
when the vendor of the device uses the mixed strategy is, after applying certain affine transformation (see equation (\ref{deltaAffine})), equivalent to the certification with a dimension witness
\begin{equation}
    W_{(\delta,y_0)}(\{0,1\},\bar{X},\bar{Y},\{\tilde{\beta}_{b,x,y}\},C_W) \nonumber
\end{equation}
with $\tilde{\beta}_{b,x,y}$ defined in the equation (\ref{deltaBeta}), and the strategy P, where the guessing probability of Eve is given by the equation (\ref{deltaAffine}).

Since the vendor may choose any $\delta \in [0,1]$ that allows to observe the required value of the dimension witness $W$ when calculating lower-bound on the certified min-entropy, the worst case should be considered for a particular situation.

This way, we have proved the following
\begin{lemma}
    Let $W = W(\{0,1\},\bar{X},\bar{Y},\{\beta_{b,x,y}\},C_W)$ be a binary zero-summing dimension witness, $x_0 \in \bar{X}$, and $y_0 \in \bar{Y}$. Let $W_{(\delta,y_0)} = W_{(\delta,y_0)}(\{0,1\},\bar{X},\bar{Y},\{\tilde{\beta}_{b,x,y}\},C_W)$ be a dimension witness, where
    \begin{equation}
        \label{deltaBeta}
        \tilde{\beta}_{b,x,y} = \left\{
        \begin{array}{l l}
            \beta_{b,x,y} & \quad \text{if $y \neq y_0$}\\
            \delta \cdot \beta_{b,x,y} & \quad \text{if $y = y_0$}
        \end{array} \right..
    \end{equation}
    Then
    \begin{equation}
        \nonumber
        \begin{aligned}
            P_{guess}^{cert}& (W,x_0,y_0,s,2) = \\
            & \max_{\delta \in [0,1]} \left( (1 - \delta) + \delta \cdot P_{guess}^{cert(P)}(W_{(\delta,y_0)},x_0,y_0,s) \right)
        \end{aligned}
    \end{equation}
    where $P_{guess}^{cert}(W,x_0,y_0,s,2)$ and $P_{guess}^{cert(P)}(W,x_0,y_0,s)$ are defined in equations (\ref{PguessCert}) and (\ref{PguessCertP}).
\end{lemma}

The consequence of restricting the vendor to the dimension $2$ and measurements of trace $1$ is that the following holds for all $x \in \bar{X}$ and $y \in \bar{Y}$, and for any $b \in \{0,1\}$:
\begin{equation}
    \label{negRho}
    \begin{aligned}
        \Tr & (\neg \rho_x M^{\neg b}_y) = \Tr((\openone - \rho_x) (\openone - M^b_y)) = \\
        & 1 - \Tr(M^b_y) + \Tr(\rho_x M^b_y) = \Tr(\rho_x M^b_y) = P(b|x,y),
    \end{aligned}
\end{equation}
where $\neg \rho_x \equiv \openone - \rho_x$. This relation allows to refine the relaxation given in the section \ref{sec:DWtoBI}.

Let us consider a device $\textbf{D1}^{\prime}$ that models the strategy P by sharing the singlet state, projecting on states $\{\rho_x\}_{x \in \bar{X}}$ on the side of Alice, and measuring on the side of Bob with measurements of trace $1$, $\{\{M^0_y,M^1_y\}\}_{y \in \bar{Y}}$. In contrast to the device $\textbf{D1}$, if the projection on a state $\rho_x$ for any $x \in \bar{X}$ fails, then the prepared state is $\neg \rho_x$. It is easy to see that, by the equation (\ref{negRho}), the probabilities obtained in this device are constrained by the following relation:
\begin{equation}
    \label{negAnegB}
    P(a,b|x,y) = P(\neg a, \neg b|x,y)
\end{equation}
for all $a, b \in \{0,1\}$, $x \in X$, and $y \in Y$. A further relaxation, analogous to the one leading from the device $\textbf{D1}$ to the device $\textbf{D2}$, allows to obtain a device $\textbf{D2}^{\prime}$, satisfying the relation (\ref{negAnegB}), that can be modeled by a semi-definite program in the device-independent scheme.

This way we have proved the following theorem
\begin{theorem}
    \label{ZeroSumTrm}
        Let $X = \bar{X}$ and $Y = \bar{Y}$ be sets. Let us take $s \in \mathbb{R}$, a Bell inequality $I$ of the form (\ref{BI}), and a binary zero-summing dimension witness $W$ of the form (\ref{DW}), satisfying $\beta_{b,x,y} = \alpha_{0,b,x,y} = \alpha_{1,\neg b,x,y}$.

    Let $\mathcal{P}_{SDI}^{(P)}(s)$ be a subset of $\mathcal{P}^{(P)}(\bar{X},\bar{Y})$ (see the definition \ref{probSDI}) containing those probabilities $\mathbb{P}_2 (\{0,1\}|\bar{X},\bar{Y})$ that satisfies $W[\mathbb{P}_2] = s$.

    Let $\mathcal{P}_{DI,cond}(s)$ be a set probability distributions defined by $P(b|x,y) \equiv P(0,b|x,y) + P(1,\neg b|x,y)$, where $\mathbb{P}(A,B|X,Y)$ is a device-independent probability distribution that satisfies $I[\mathbb{P}(A,B|X,Y)] = s$ and the relation (\ref{negAnegB}).

    Then $\mathcal{P}_{SDI}^{(P)}(s) \subseteq \mathcal{P}_{DI,cond}(s)$.
\end{theorem}

It is straightforward to check that the following lemma holds:
\begin{lemma}
    \label{probImply}
    Let $a, b \in \{0, 1\}$, and let us assume that
    \begin{subequations}
        \begin{equation}
            P(a,b|x,y)+P(a,\neg b|x,y)=P(a|x,y)=P(a|x), \text{and} \nonumber
        \end{equation}
        \begin{equation}
            P(a,b|x,y)+P(\neg a,b|x,y)=P(b|x,y)=P(b|y), \nonumber
        \end{equation}
    \end{subequations}
    \textit{i.e.} the no-signaling principle, and that the outcomes of Bob are binary, namely
    \begin{equation}
        P(b|a,x,y)+P(\neg b|a,x,y)=1.
    \end{equation}
    Then we have the following implications:
    \begin{enumerate}
        \item If $P(a,b|x,y)=P(\neg a, \neg b|x,y)$ holds, then we have $P(a|x)=\frac{1}{2}$ and $P(b|a,x,y)+P(b|\neg a,x,y)=1$.
        \item If $P(b|a,x,y)+P(b|\neg a,x,y)=1$ holds, then we have $P(b|a,x,y)=P(\neg b|\neg a,x,y)$.
        \item If $P(a|x)=\frac{1}{2}$ and $P(b|a,x,y)=P(\neg b|\neg a,x,y)$ hold, then we have $P(a,b|x,y)=P(\neg a,\neg b|x,y)$.
    \end{enumerate}
\end{lemma}
From this lemma we get that the condition $P(a,b|x,y)=P(\neg a, \neg b|x,y)$ is more restrictive than $P(a|x)=\frac{1}{2}$. From this we conjecture that for any $s$
\begin{equation}
    \mathbb{P}_{SDI}^{(P)}(s) \subseteq \mathbb{P}_{DI,cond}(s) \subseteq \mathbb{P}_{DI}(s), \nonumber
\end{equation}
where the sets are defined in theorems \ref{DWtoBItheorem} and \ref{ZeroSumTrm}. Thus the theorem \ref{ZeroSumTrm} refines the results of the theorem \ref{DWtoBItheorem} for the case of binary zero-summing dimension witnesses.

\begin{figure*}[!htbp]
        \centering

        \subfloat[T2]{\includegraphics[width=0.45\textwidth]{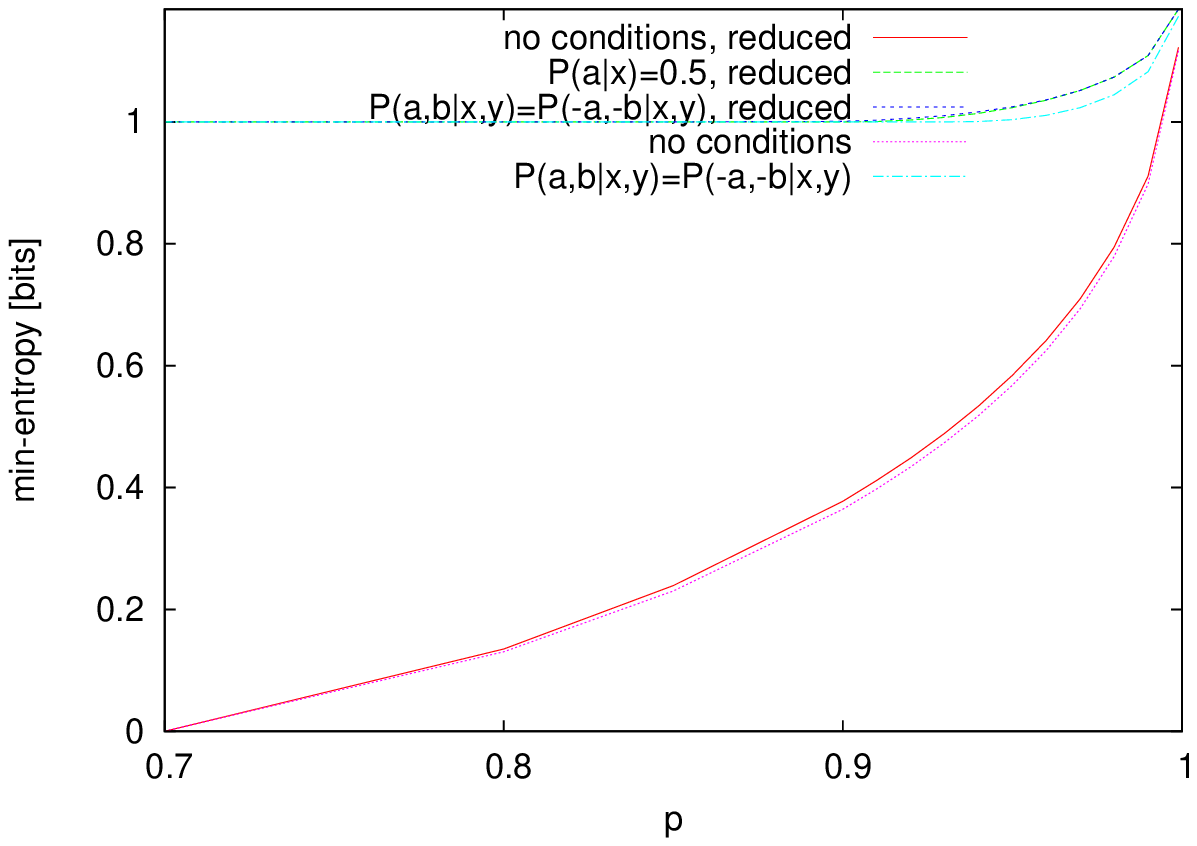}\label{fig:T2_NPA_P}}
        \subfloat[T3]{\includegraphics[width=0.45\textwidth]{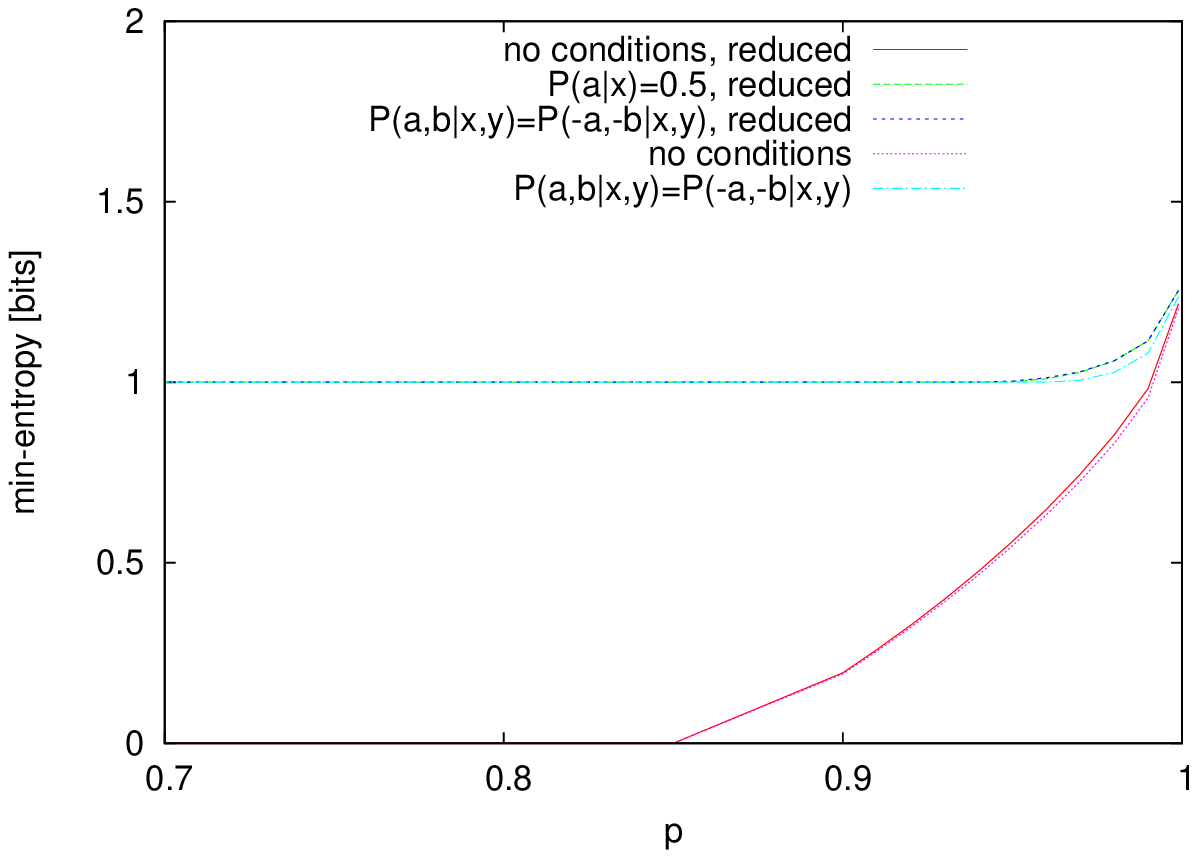}\label{fig:T3_NPA_P}}

        \subfloat[BC3]{\includegraphics[width=0.45\textwidth]{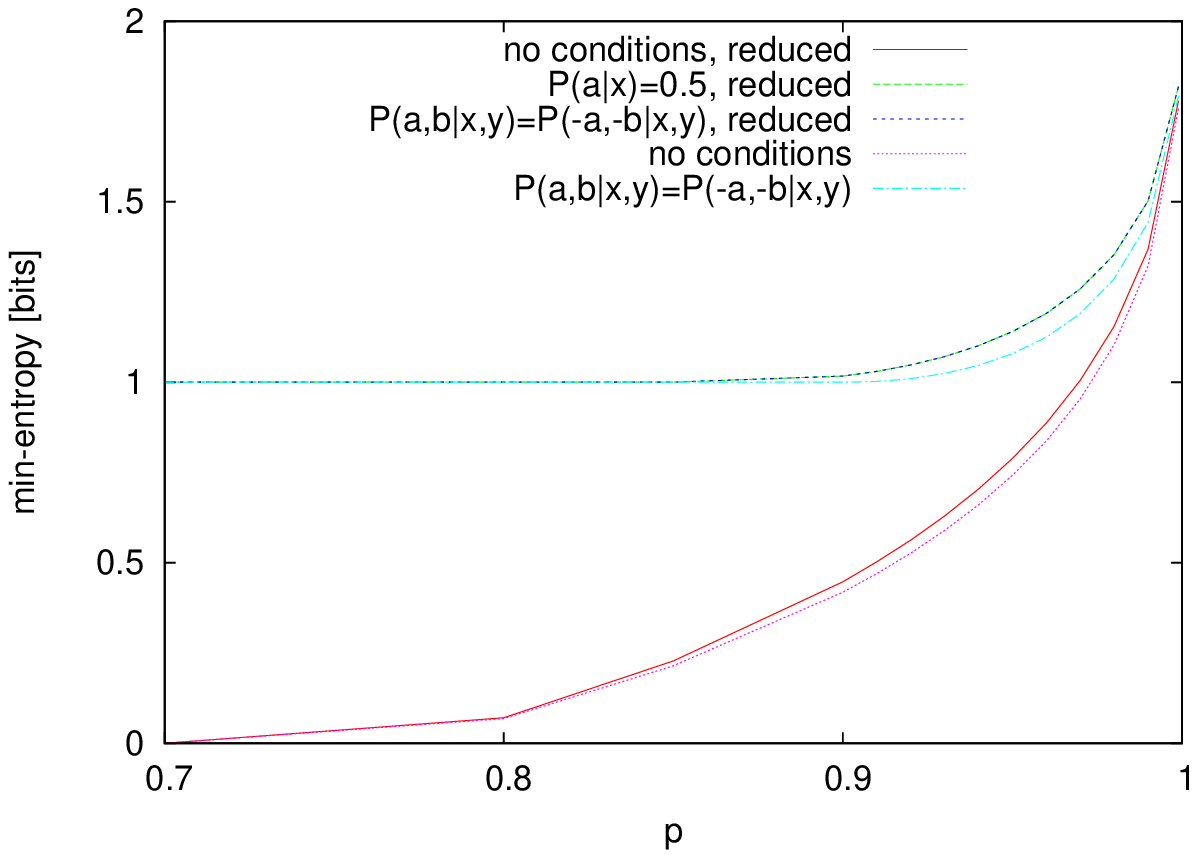}\label{fig:BC3_NPA_P}}
        \subfloat[modCHSH]{\includegraphics[width=0.45\textwidth]{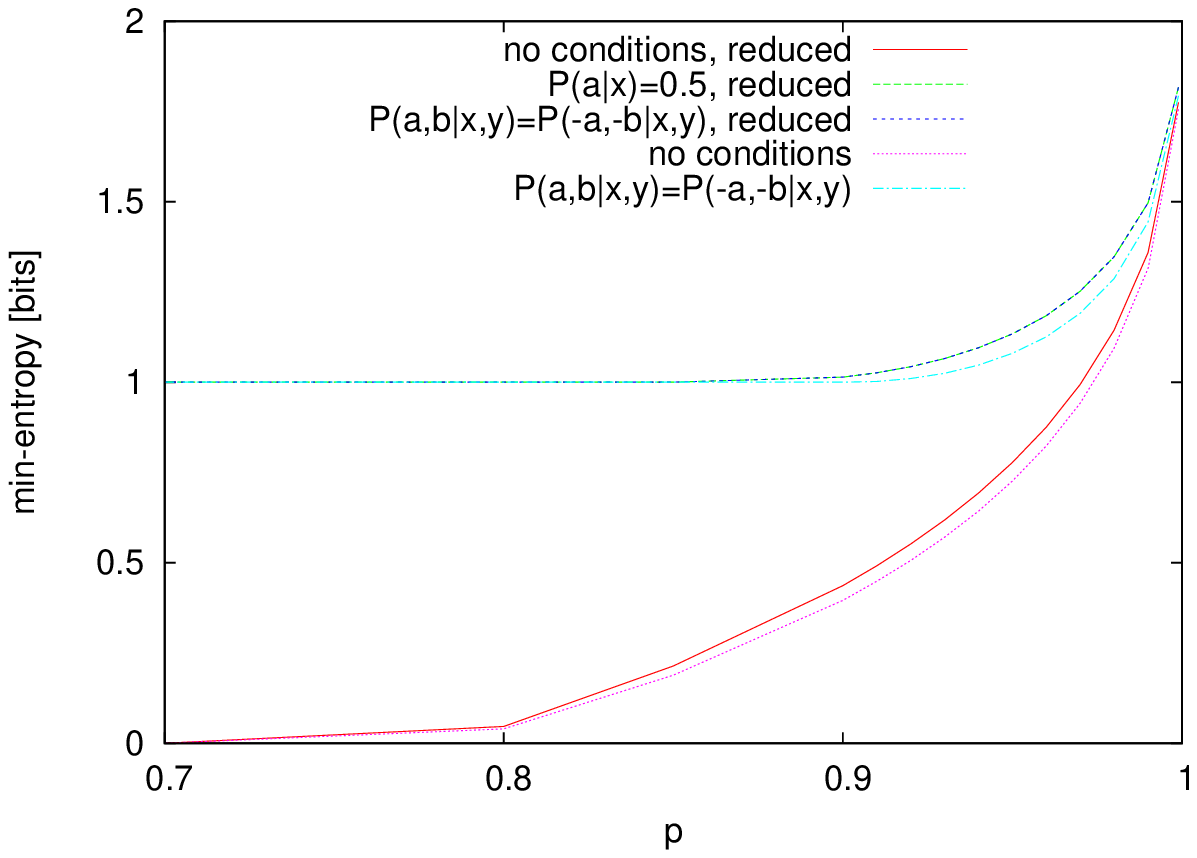}\label{fig:modCHSH_NPA_P}}

        \caption{(Color on-line) Lower-bounds via SDP on min-entropy certified by different Bell inequalities for different levels of noise. Several situations are considered. This illustrates the relations summarized in the lemma \ref{probImply}. \textit{Reduced} Bell inequalities refer to lower-bounds on reduced symmetric dimension witnesses (see the section \ref{sec:symDW}) for the strategy P (see the section \ref{sec:binaryDW}). These Bell operators are given by formulas (\ref{T2BI}), (\ref{T3BI}), (\ref{BC3BI}) and (\ref{modCHSHBI}). We observe that using reduced dimension witnesses provides an advantage in the terms of certifiable randomness. Bell operators that may be used for lower-bounding the randomness full symmetric dimension witnesses are given by formulas (\ref{T2BIfull}), (\ref{T3BIfull}), (\ref{BC3BIfull}) and (\ref{modCHSHBIfull}). Recall that the theorem \ref{DWtoBItheorem} and a remark below it, the lower-bound for the randomness of a dimension witness is given by the randomness of the Bell inequality minus $\log_2{d}$. This plot refers to the case with $d = 2$, thus these methods that give value below $\log_2{2} = 1$ are not feasible for the given value of $p$. \label{fig:PStrategyLower}}
\end{figure*}

\begin{figure*}[!htbp]
        \centering

        \subfloat[T2]{\includegraphics[width=0.45\textwidth]{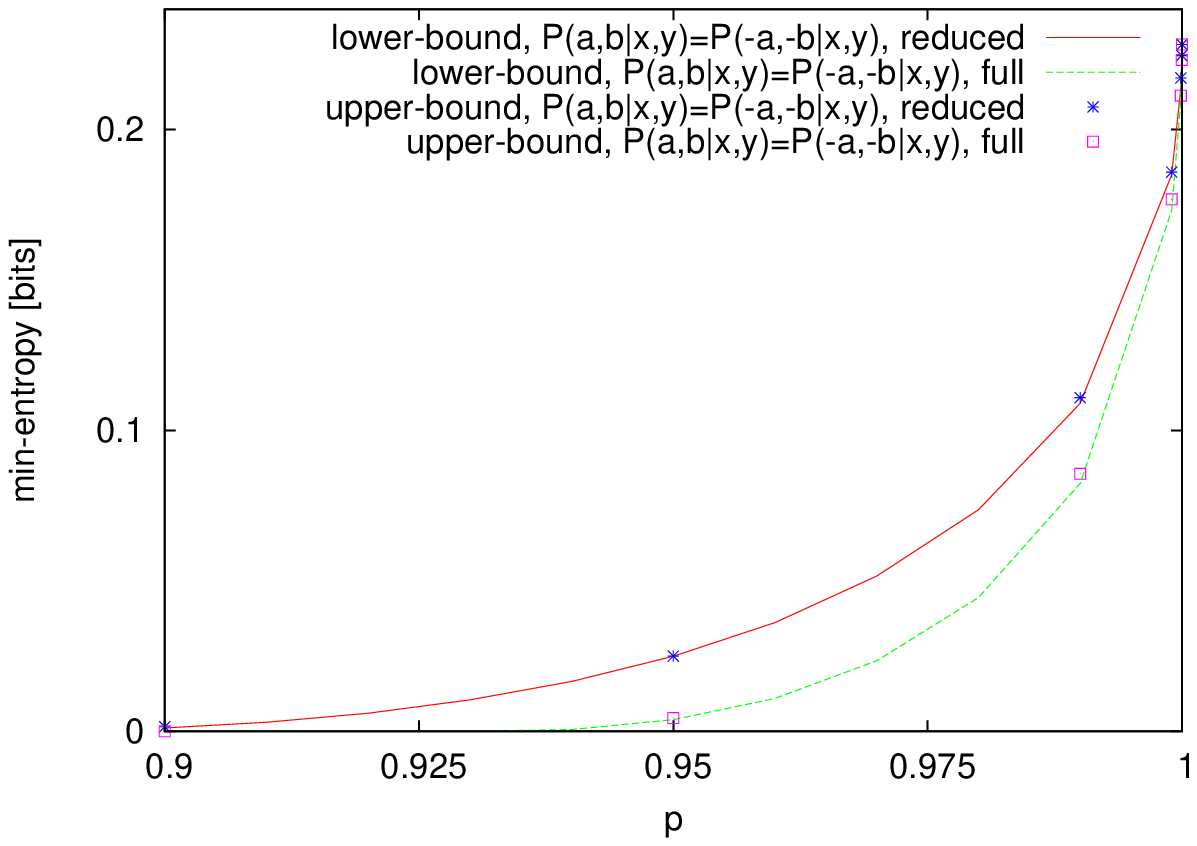}\label{fig:T2_cert_P}}
        \subfloat[T3]{\includegraphics[width=0.45\textwidth]{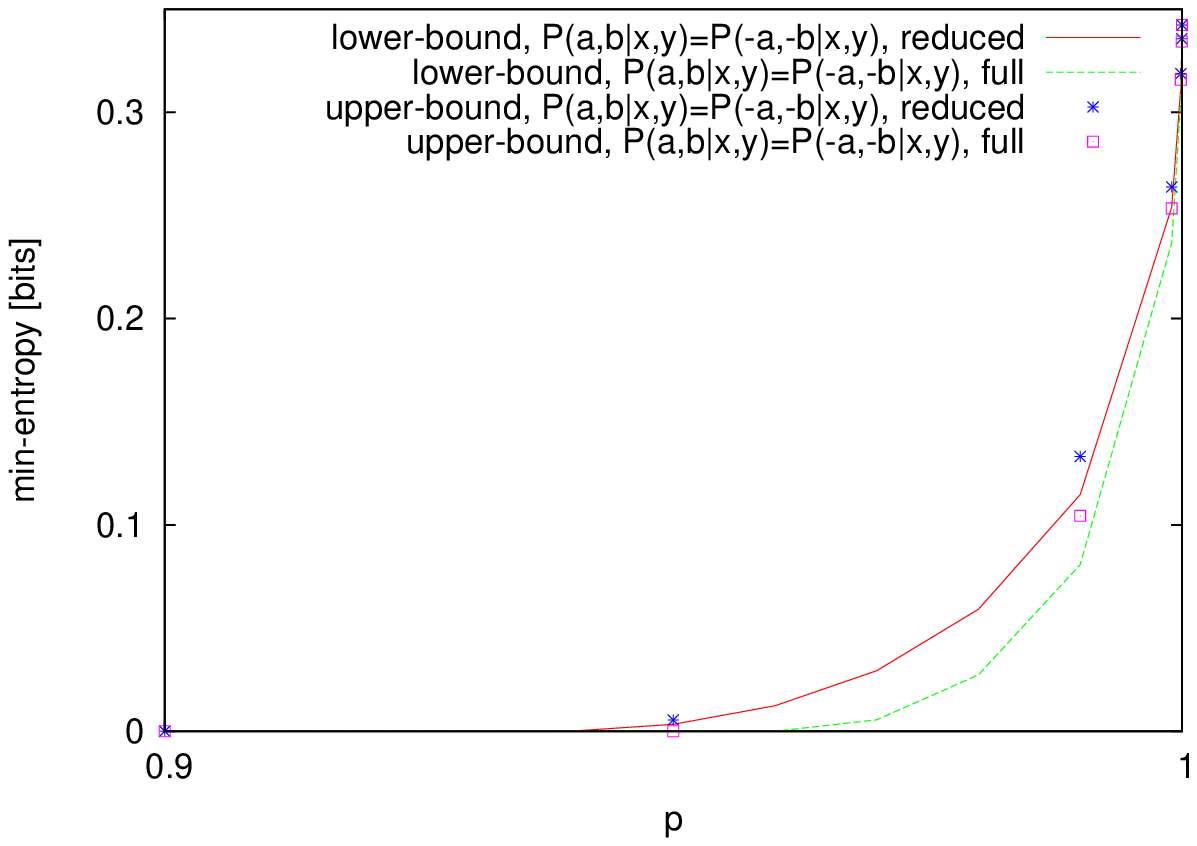}\label{fig:T3_cert_P}}

        \subfloat[BC3]{\includegraphics[width=0.45\textwidth]{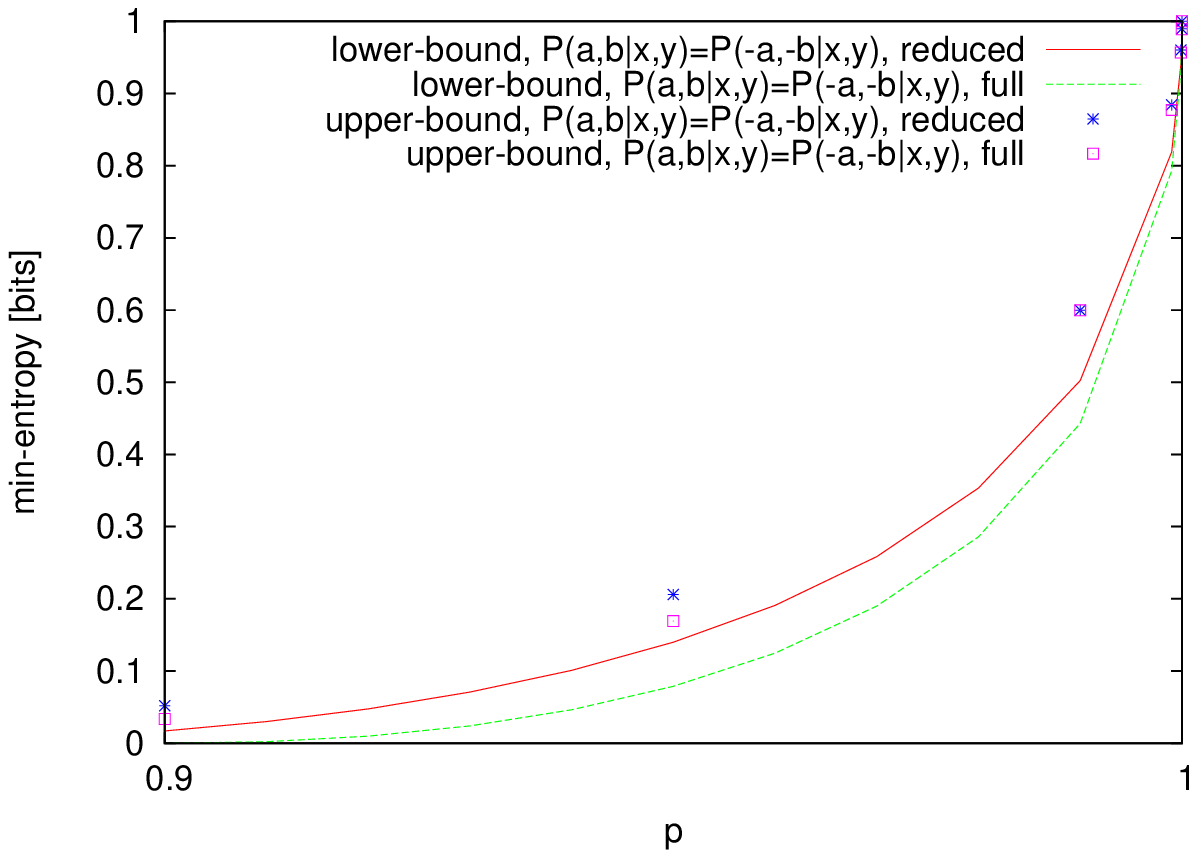}\label{fig:BC3_cert_P}}
        \subfloat[modCHSH]{\includegraphics[width=0.45\textwidth]{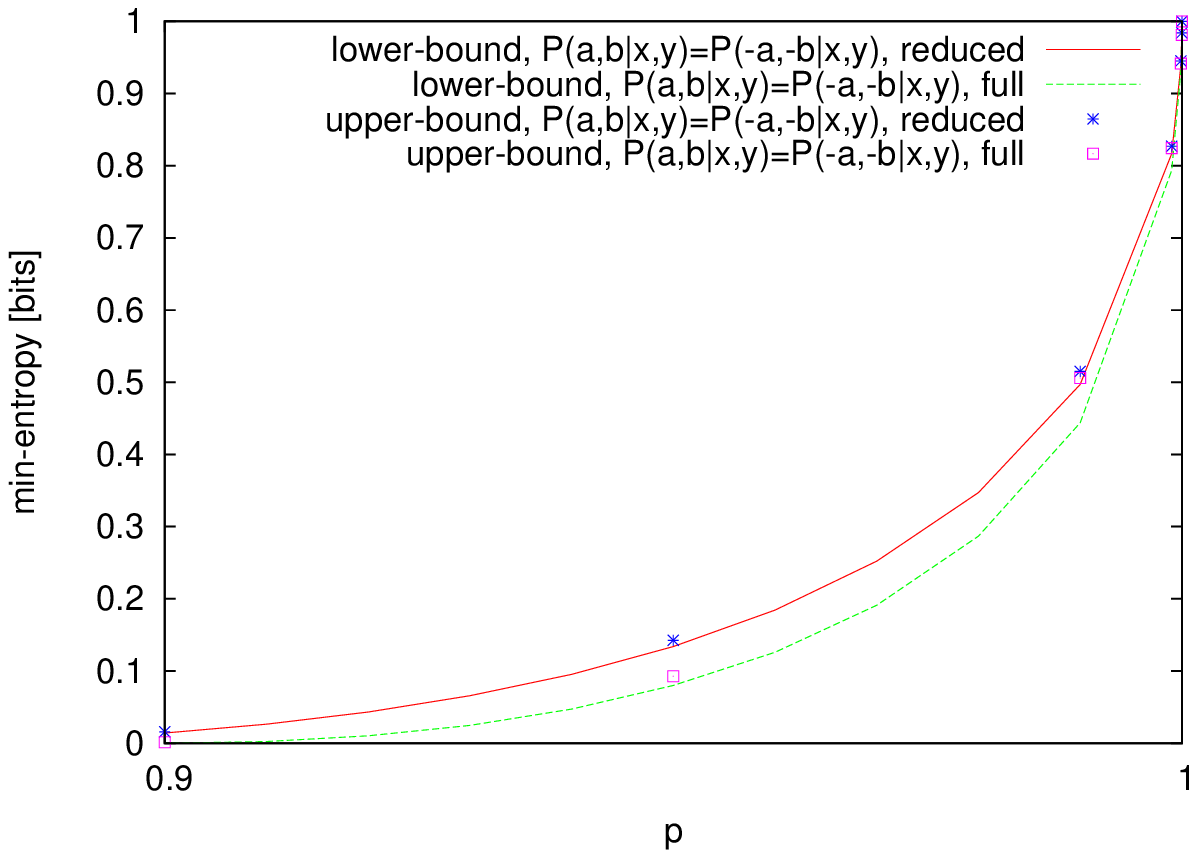}\label{fig:modCHSH_cert_P}}

        \caption{(Color on-line) Numerical lower-bounds via SDP and upper-bounds on the randomness certified for the strategy P (see the section \ref{sec:binaryDW}), for both reduced (see the section \ref{sec:symDW}) and full certificates and different Bell operators and dimension witnesses. Here we observe even more advantage for the reduced versions.\label{fig:PStrategyLowerUpper}}
\end{figure*}

\begin{figure*}[!htbp]
        \centering

        \subfloat[T2]{\includegraphics[width=0.45\textwidth]{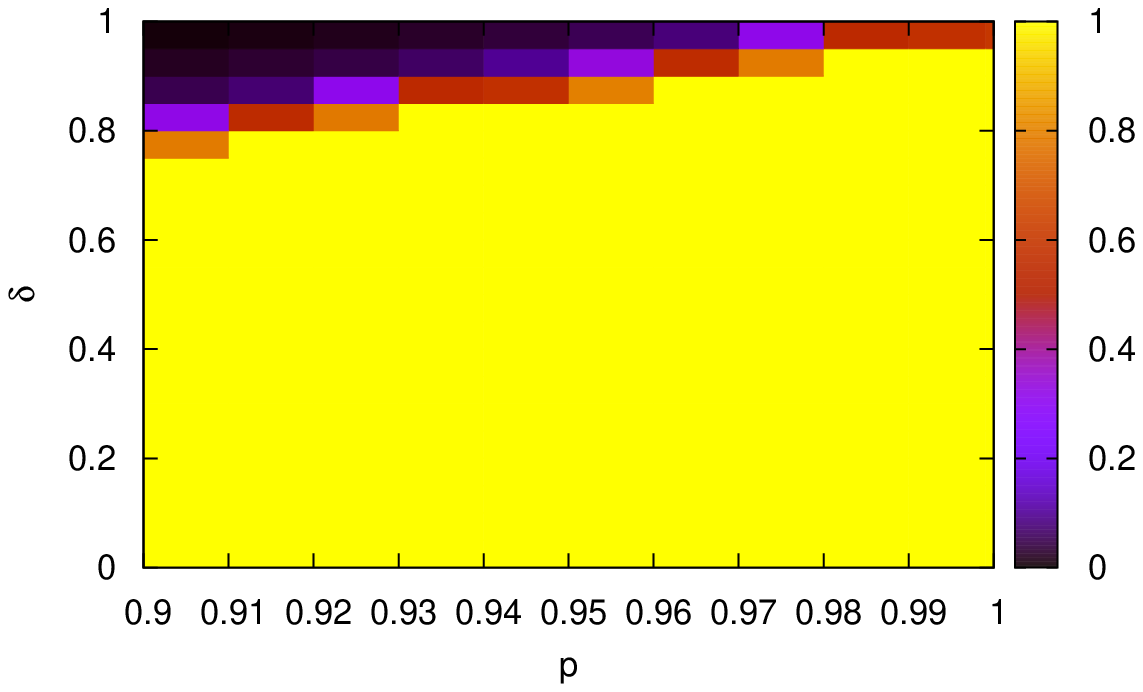}\label{fig:T2_NPA_delta}}
        \subfloat[T3]{\includegraphics[width=0.45\textwidth]{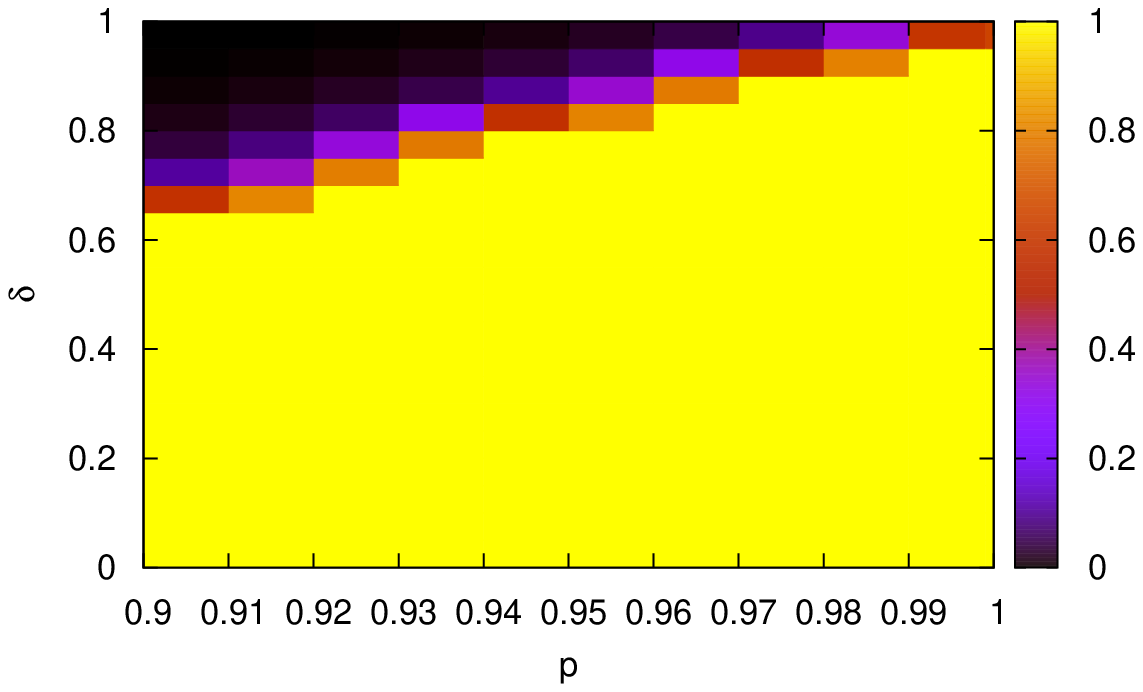}\label{fig:T3_NPA_delta}}

        \subfloat[BC3]{\includegraphics[width=0.45\textwidth]{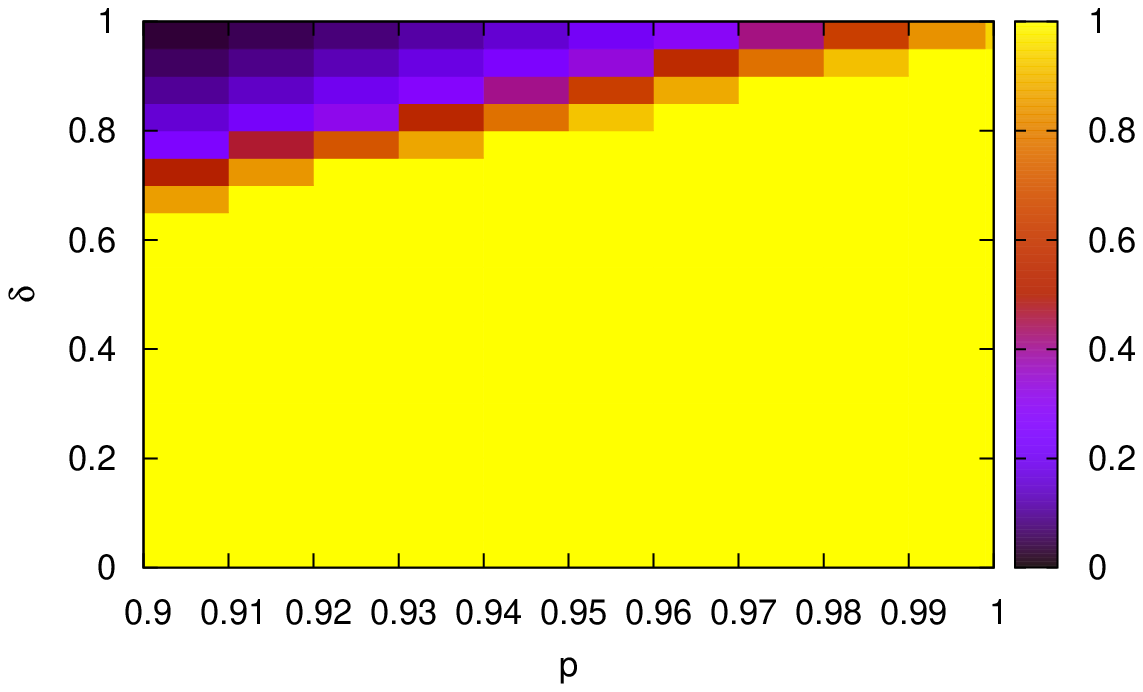}\label{fig:BC3_NPA_delta}}
        \subfloat[modCHSH]{\includegraphics[width=0.45\textwidth]{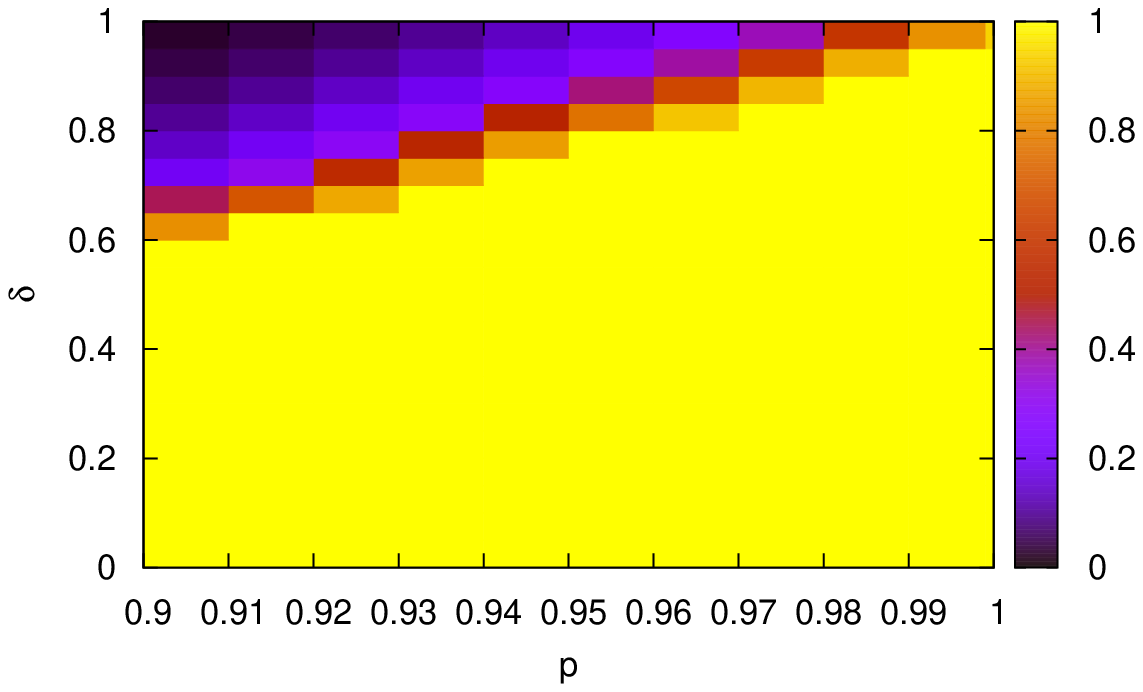}\label{fig:modCHSH_NPA_delta}}

        \caption{(Color on-line) Lower-bounds via SDP on the certified randomness in semi-device-independent scenario for both reduced dimension witnesses (see the section \ref{sec:symDW}) when the untrusted vendor uses the mixed strategy with different values of the parameter $\delta$ (see the section \ref{sec:binaryDW}). If certain value of a dimension witness is impossible to be achieved with given $\delta$, then, since the eavesdropper cannot mislead us this way, the value $1$ is put.\label{fig:MixedStrategyLowerDelta}}
\end{figure*}

\begin{figure*}[!htbp]
        \centering

        \subfloat[T2]{\includegraphics[width=0.45\textwidth]{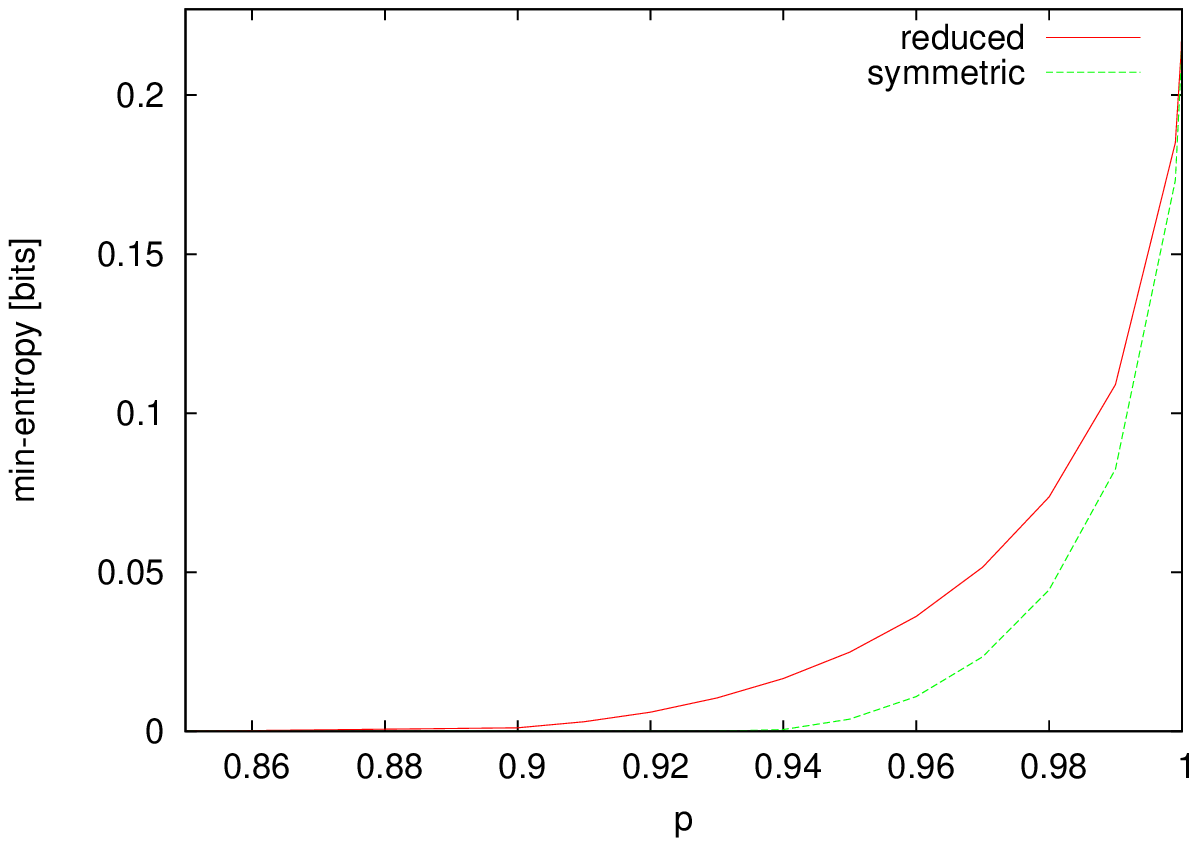}\label{fig:T2_cert}}
        \subfloat[T3]{\includegraphics[width=0.45\textwidth]{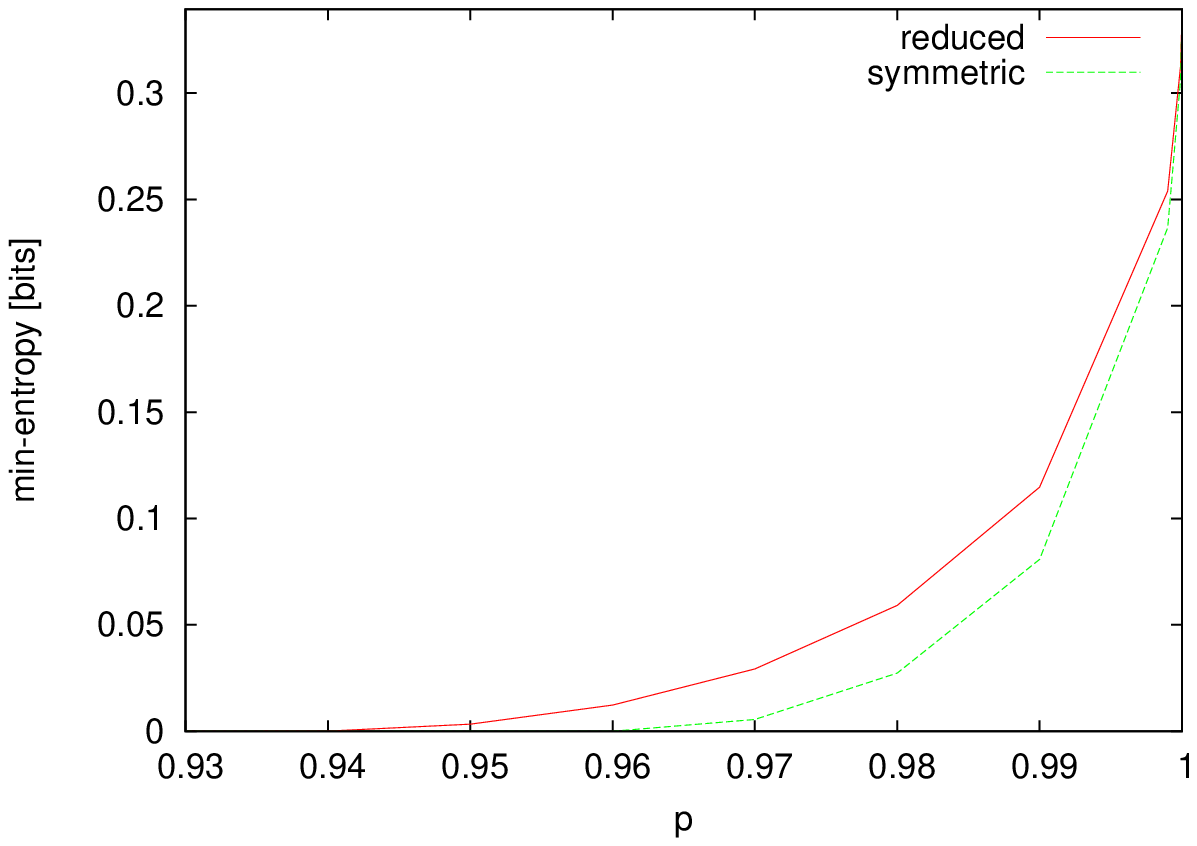}\label{fig:T3_cert}}

        \subfloat[BC3]{\includegraphics[width=0.45\textwidth]{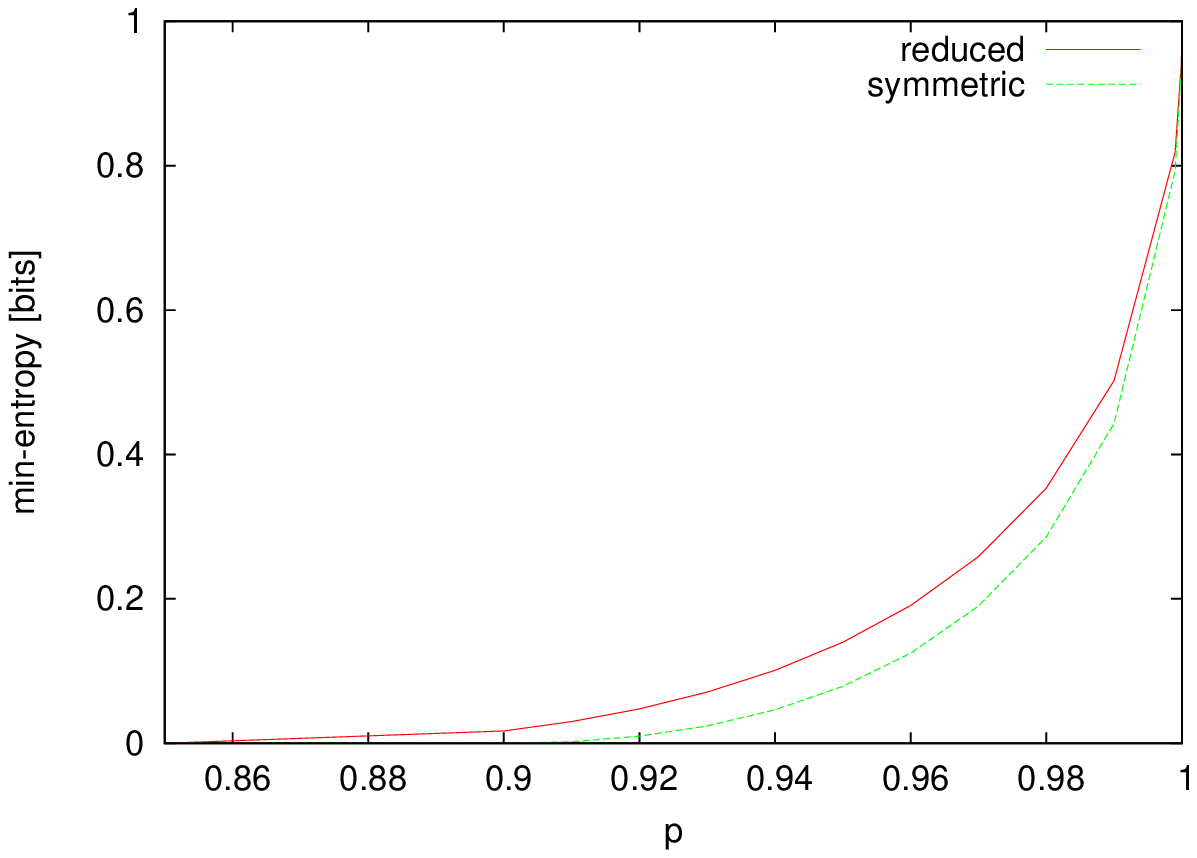}\label{fig:BC3_cert}}
        \subfloat[modCHSH]{\includegraphics[width=0.45\textwidth]{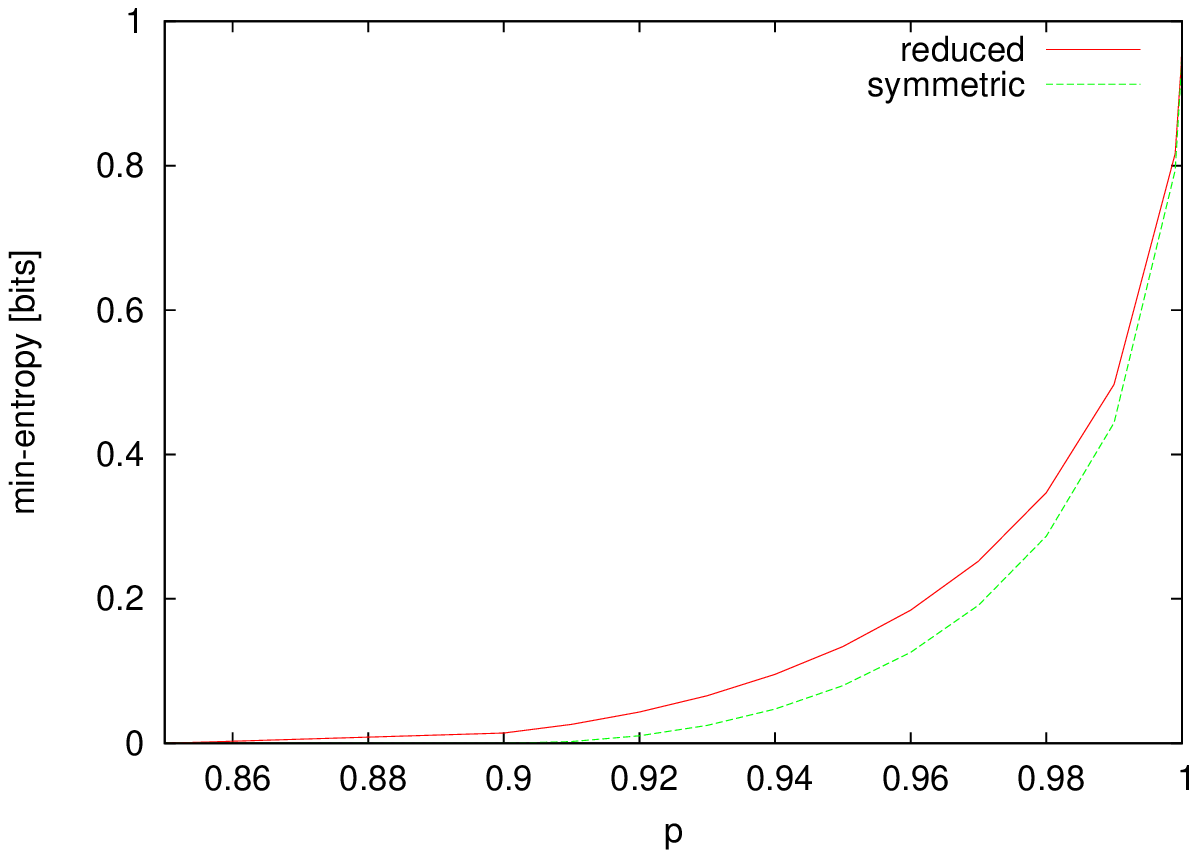}\label{fig:modCHSH_cert}}

        \caption{(Color on-line) Lower-bounds via SDP on the certified randomness for both reduced and symmetric dimension witnesses (see the section \ref{sec:symDW}) when the untrusted vendor uses the mixed strategy with the optimal value of the parameter $\delta$ (see the section \ref{sec:binaryDW}) which occurs to be $1$ in all cases.\label{fig:MixedStrategyLower}}
\end{figure*}

Figures \ref{fig:PStrategyLower} and \ref{fig:PStrategyLowerUpper} show examples lower- and upper-bounds for min-entropy certified when the untrusted vendor uses the strategy P. Figures \ref{fig:MixedStrategyLowerDelta} and \ref{fig:MixedStrategyLower} show lower-bounds for the certified min-entropy in case when the untrusted vendor uses the mixed strategy. All lower-bounds are calculated via semi-definite programs with the NPA hierarchy, using interior point method with the SeDuMi solver \cite{SeDuMi102,IntPoint}. The upper-bounds have been obtained by finding explicit representations of states and measurements. This optimization has been carried over pure states and projective measurements, and is not guarantied to reach global minima, in contrast to the semi-definite programming method.

Interestingly, in all protocols considered in the figure \ref{fig:MixedStrategyLower}, it is optimal for the adversary to use $\delta = 1$, \textit{i.e.} using the mixed strategy gives no gain comparing to the strategy P.

\section{Symmetric dimension witnesses}
\label{sec:symDW}

Let us introduce the following definition:
\begin{definition}
    \label{symDWdef}
    A dimension witness $W$ of the form (\ref{DW}) with the set of Alice's settings $\bar{X}$ of even size, and $\bar{B} = \{0,1\}$, is \textit{symmetric}, if there exist an surjective automorphism $\phi: \bar{X} \rightarrow \bar{X}$ with $\phi(x) \neq x$ and $\beta_{b,x,y} = -\beta_{b,\phi(x),y} = -\beta_{\neg b,x,y}$.

    For a set $\bar{\chi} \subset \bar{X}$ we define
    \begin{equation}
        W_{\bar{\chi}} \equiv \sum_{b \in \{0,1\}} \sum_{x \in \bar{\chi}} \sum_{y \in \bar{Y}} \beta_{b,x,y} P(b|x,y). \nonumber
    \end{equation}

    A set $\bar{\chi} \subset \bar{X}$ satisfying $\bar{\chi} \cap \phi(\bar{\chi}) = \emptyset$, and $\bar{\chi} \cup \phi(\bar{\chi}) = \bar{X}$ is called a half of $\bar{X}$.

    If a set $\bar{\chi}$ is a half, then $W_{\bar{\chi}}$ is called a dimension witness reduced with respect to $\chi$. $\phi$ and $\chi$ may be omitted if it is obvious which automorphism or set is considered.
\end{definition}

If a dimension witness is symmetric, then there is a way to reduce the size of $\bar{X}$, whilst the obtained dimension witness can certify at least the same amount of randomness, as the initial one.

The following theorem is an immediate result of the theorem \ref{ZeroSumTrm} and the lemma \ref{dwLemma}.2:
\begin{theorem}
    \label{symAddCond}
    For a SDI protocol using the strategy P with a symmetric dimension witness that attains the value of the security parameter $s$ on a Hilbert space of the dimension $2$ and certifies the randomness $r$, the same value is still possible to be attained and certifies at least the same randomness, if we impose an additional condition that $\rho_x = \openone - \rho_{\phi(x)}$, which implies $P(b|x,y)=P(\neg b|\phi(x),y)$.
\end{theorem}

Simply speaking this theorem says that symmetric dimension witnesses posses some kind of degree of freedom that does not increase the range of values possible to be attained, but allows adversary to ``distribute'' the value of the witness among the states in such a way that misleads about the reliability of the device. The proposed method shows a way to remove this freedom.

\subsection{Obtaining and reduction of a symmetric dimension witness}

This subsection shows how to transform a symmetric dimension witness to a reduced one.

It is possible to use a known Bell inequality to obtain a new dimension witness. Examples of such protocols, $T2$, $T3$, $BC3$ and $modCHSH$, are described below in the section \ref{sec:explicitExamples}. From a Bell inequality of the form
\begin{equation}
    \label{symBI}
    \begin{aligned}
        \sum_{x,y} \dot{\alpha}_{x,y} &\left( \frac{1}{2 p_{0,x}} \left( P(0,0|x,y) - P(0,1|x,y) \right) \right.\\
            + &\left. \frac{1}{2 p_{1,x}} \left( P(1,1|x,y) - P(1,0|x,y) \right) \right)
    \end{aligned}
\end{equation}
(where $p_{0,x}+p_{1,x}=1$), using the method from the section \ref{sec:BItoDW}, we obtain a symmetric dimension witness of the form (\ref{DW}), with $\beta_{0,x,y} = \dot{\alpha}_{x,y}$, and $\beta_{1,x,y} = -\dot{\alpha}_{x,y}$. For the new SDI protocol, we assume that $a$ is chosen randomly by Alice, with the distribution $P(a|x) \equiv p_{a,x}$.

Note that Bell inequalities in the correlation form (see the equation (\ref{BIhat})), are a special case of the inequalities of the form (\ref{symBI}) with $p_{a,x} = \frac{1}{2}$ and $\dot{\alpha}_{x,y} = \bar{\alpha}_{x,y}$, which means that it is always possible to obtain a symmetric dimension witness from a correlation-based Bell inequality. Then $C(x,y)$ turns into
\begin{equation}
    \label{CtoWfull}
    \begin{aligned}
        W^{\prime}(x, y) &\equiv \frac{1}{2} \left( P(0|(0,x),y)+P(1|(1,x),y) \right. \\
        &\left. -P(1|(0,x),y)-P(0|(1,x),y) \right) \\
        &= P(0|(0,x),y) - P(0|(1,x),y).
    \end{aligned}
\end{equation}
It is easy to see, that a dimension witness which is a linear combination of expressions (\ref{CtoWfull}), is symmetric.

We define $\phi((a,x)) \equiv (\neg a,x)$ and $\chi \equiv \{(0,x): x \in X\} \subseteq \{0,1\} \times X$. The condition $P(b|(a,x),y)=P(\neg b|(\neg a,x),y)$ allows us to take
\begin{equation}
    \label{CtoW}
    W(x, y) \equiv 2 P(0|(0,x),y) - 1 \equiv 2 P(0|x,y) - 1,
\end{equation}
instead of $W^{\prime}(x, y)$ from the equation (\ref{CtoWfull}), which is an example of the reduction.

Note that using the method of reduction of a symmetric dimension witness the number of states used by Alice is reduced twice without loss of of ability to certify both the randomness and the dimension.

On the other hand every symmetric dimension witness is a linear combination of expressions $D(x,y) \equiv P(0|x,y) - P(1|x,y)$ that refers in DI scenario to the expression $2 P(0,0|x,y) - 2 P(0,1|x,y)$. Assuming that the dimension of the Hilbert space is $2$, and the eavesdropper uses the strategy P, we get from the equation (\ref{negAnegB}) that $2 P(0,0|x,y) - 2 P(0,1|x,y) = C(x,y)$.

\section{Examples}
\label{sec:explicitExamples}

In this section we give five examples of applications of the methods presented above. Four of them, B, C, D, and E, concern Bell inequalities in the correlation form and symmetric dimension witnesses.

All figures are plotted with respect to a relative parameter $p$. The value $p = 1$ refers the case when the maximal value of the relevant Bell inequality or dimension witness is achieved. Values $p < 1$ relate to the situation with noise, when the attained value is equal to the maximum multiplied by $p$.

\subsection{CGLMP}
\label{sub:CGLMP}

In the first example we start with CGLMP inequality introduced in \cite{CGLMP}. Both Alice and Bob have two measurement settings with three outcomes. It has the following form
\begin{equation}
    \begin{aligned}
        &  P(0,0|1,1)-P(0,2|1,1)+P(0,0|1,2)-P(0,2|1,2) \\
        & -P(1,0|1,1)+P(1,1|1,1)-P(1,0|1,2)+P(1,1|1,2) \\
        & -P(2,1|1,1)+P(2,2|1,1)-P(2,1|1,2)+P(2,2|1,2) \\
        & -P(0,0|2,1)+P(0,1|2,1)+P(0,0|2,2)-P(0,2|2,2) \\
        & -P(1,1|2,1)+P(1,2|2,1)-P(1,0|2,2)+P(1,1|2,2) \\
        & +P(2,0|2,1)-P(2,2|2,1)-P(2,1|2,2)+P(2,2|2,2). \nonumber
    \end{aligned}
\end{equation}
Using the heuristic method from the section \ref{sec:BItoDW} we obtain the following dimension witness:
\begin{equation}
    \label{CGLMPdw}
    \begin{aligned}
        &  P(0|1,1)-P(2|1,1)+P(0|1,2)-P(2|1,2) \\
        & -P(0|2,1)+P(1|2,1)-P(0|2,2)+P(1|2,2) \\
        & -P(1|3,1)+P(2|3,1)-P(1|3,2)+P(2|3,2) \\
        & -P(0|4,1)+P(1|4,1)+P(0|4,2)-P(2|4,2) \\
        & -P(1|5,1)+P(2|5,1)-P(0|5,2)+P(1|5,2) \\
        & +P(0|6,1)-P(2|6,1)-P(1|6,2)+P(2|6,2).
    \end{aligned}
\end{equation}
Applying the method from the section \ref{sec:DWtoBI}, we get the following expression, which may be used in a semi-definite program:
\begin{equation}
    \begin{aligned}
        &  P(0,0|1,1)-P(0,2|1,1)+P(0,0|1,2)-P(0,2|1,2) \\
        & -P(0,0|2,1)+P(0,1|2,1)-P(0,0|2,2)+P(0,1|2,2) \\
        & -P(0,1|3,1)+P(0,2|3,1)-P(0,1|3,2)+P(0,2|3,2) \\
        & -P(0,0|4,1)+P(0,1|4,1)+P(0,0|4,2)-P(0,2|4,2) \\
        & -P(0,1|5,1)+P(0,2|5,1)-P(0,0|5,2)+P(0,1|5,2) \\
        & +P(0,0|6,1)-P(0,2|6,1)-P(0,1|6,2)+P(0,2|6,2). \nonumber
    \end{aligned}
\end{equation}
The certified randomness for CGLMP is shown in the figure \ref{fig:CGLMP}.

\subsection{BC3}
\label{sub:BC3}

In the second example we start with a well known Braunstein-Caves inequality (denoted below $BC3$, it is a Bell inequality in the form (\ref{BIhat})) with three settings for each of the two parties, and convert it to a symmetric dimension witness with six prepared states. After reduction, we will obtain a dimension witness with three states, and show that the lower bounding Bell inequality is identical to the original $BC3$. 

$BC3$ inequality is of the form
\begin{equation}
    \label{BC3BI}
    \begin{aligned}
        BC_{3} & \equiv C(1, 1) + C(1, 2) + C(2, 2) \\
        & + C(2, 3) + C(3, 3) - C(3, 1),
    \end{aligned}
\end{equation}
with $\delta = 1$. For $BC_{3}$ we have $x, y \in \{1, 2, 3\}$. Thus we obtain a symmetric dimension witness with six states prepared by Alice and three measurements performed by Bob.

The explicit form of this symmetric dimension witness is
\begin{equation}
	\begin{aligned}
		& P(0|(0,1),1) - P(0|(1,1),1) + P(0|(0,1),2) \\
		& - P(0|(1,1),2) + P(0|(0,2),2) - P(0|(1,2),2) \\
		& + P(0|(0,2),3) - P(0|(1,2),3) + P(0|(0,3),3) \\
		& - P(0|(1,3),3)	- P(0|(0,3),1) + P(0|(1,3),1).
	\end{aligned} \nonumber
\end{equation}

Using the method for symmetric dimension witnesses from the section \ref{sec:symDW}, this may be transformed into a dimension witness with three states. We define $\phi((a,x)) \equiv (\neg a,x)$ and $\chi \equiv \{(0,x): x \in X\}  \subseteq \{0,1\} \times X$.

The explicit form of this reduced dimension witness is
\begin{equation}
	\begin{aligned}
		2 \left( P(0|1,1) + P(0|1,2) + P(0|2,2) \right. \\
		\left. + P(0|2,3) + P(0|3,3) - P(0|3,1) \right) - 4.
	\end{aligned} \nonumber
\end{equation}

Now, using the theorem \ref{ZeroSumTrm}, we go from this reduced dimension witness back to the Bell inequality that gives a lower-bounding relation. Assuming $P(a,b|x,y)=P(\neg a, \neg b|x,y)$, we get that the lower-bounding Bell inequality is exactly the initial Braunstein-Caves inequality. If we use a full dimension witness, then the Bell inequality used in lower-bounding with the theorem \ref{ZeroSumTrm} is
\begin{equation}
	\label{BC3BIfull}
	\begin{aligned}
		& \frac{1}{2} \left( C(1, 1) + C(1, 2) + C(2, 2) + C(2, 3) \right. \\
		& + C(3, 3) - C(3, 1) + C(4, 1) + C(4, 2) \\
		&\left. + C(5, 2) + C(5, 3) + C(6, 3) - C(6, 1) \right),
	\end{aligned}
\end{equation}
where $(0, 1) \equiv 1$, $(0, 2) \equiv 2$, $(0, 3) \equiv 3$, $(1, 1) \equiv 4$, $(1, 2) \equiv 5$, and $(1, 3) \equiv 6$.

In the figure \ref{fig:BC3_NPA_P}, the min-entropies certified with the Bell inequality $BC3$ with different additional conditions are plotted. The figure \ref{fig:BC3_cert_P} shows lower-bounds on the min-entropy certified in this SDI protocol, obtained by theorem \ref{DWtoBItheorem} from the NPA hierarchy with additional condition $P(a,b|x,y)=P(\neg a,\neg b|x,y)$. These values assume that the untrusted vendor uses the strategy P (see the section \ref{sec:binaryDW}). Plots relevant to the mixed strategy are shown on figures \ref{fig:BC3_NPA_delta} and \ref{fig:BC3_cert}.

\subsection{T3}
\label{sub:T3}

The third example starts with a dimension witness based on $3$ to $1$ quantum random access code \cite{QRAC,DW4}, and relates it, and its reduced version, to two Bell inequalities, where the second one is $T3$ introduced in \cite{HWL12}.

In $3$ to $1$ quantum random access code Alice encodes three bits by sending one of the $2^{3}$ states to Bob, who tries to guess one of them, performing one of three measurements. The average success probability of correctly guessing an arbitrarily chosen bit is directly related to the value of the following dimension witness:
\begin{equation}
	\label{T3DWfull}
	\sum_{x \in \bar{X}, y \in \bar{Y}} (-1)^{x_y} P(0|x,y),
\end{equation}
where $\bar{X} = \{000, \ldots, 111\}$, $\bar{Y} = \{0,1,2\}$. Its maximal value attainable with qubits is $4 \sqrt{3}$.

Taking $\phi(x)=\neg x$ (negation is meant here as bit-wise), $\bar{X} = \{00,01,10,11\}$ and $\bar{Y} = \bar{Y}$, we get the following reduced dimension witness
\begin{equation}
	\label{T3DW}
	\begin{aligned}
		& P(0|00,0)+P(0|01,0)+P(0|10,0)+P(0|11,0) \\
		& +P(0|00,1) +P(0|00,2) +P(0|01,1) -P(0|01,2) \\
		& -P(0|10,1) +P(0|10,2) -P(0|11,1) -P(0|11,2).
	\end{aligned}
\end{equation}
From this dimension witness, using the method from the section \ref{sec:symDW}, we get the following Bell operator:
\begin{equation}
	\label{T3BI}
	\begin{aligned}
		& C(1,1)+C(2,1)+C(3,1)+C(4,1) \\
		& +C(1,2) +C(1,3) +C(2,2) -C(2,3) \\
		& -C(3,2) +C(3,3) -C(4,2) -C(4,3).
	\end{aligned}
\end{equation}
If we do not reduce the dimension witness and use the formula (\ref{T3DWfull}) directly, we get the following Bell operator:
\begin{equation}
    \label{T3BIfull}
    \begin{aligned}
        T3^{\prime} \equiv & \frac{1}{2} \left( C(1,1)-C(5,1) +C(2,1)-C(6,1) \right. \\
        & +C(3,1)-C(7,1) +C(4,1)-C(8,1) \\
        & +C(1,2)-C(5,2) +C(1,3)-C(5,3) \\
        & +C(2,2)-C(6,2) -C(2,3)+C(6,3) \\
        & -C(3,2)+C(7,2) +C(3,3)-C(7,3) \\
        & \left. -C(4,2)+C(8,2) -C(4,3)+C(8,3) \right).
    \end{aligned}
\end{equation}
The Bell operator defined in the equation (\ref{T3BI}) is the one used in \cite{HWL12, HWL13, MP13}.

It is possible to calculate a lower-bound on the certified min-entropy, $H_{\infty}^{cert}(T3,x_0,y_0,s,d)$, with $x_0 = 1$, $y_0 = 1$ using the theorem \ref{ZeroSumTrm}, \textit{i.e.} via a semi-definite relaxation with a minimization on higher level over $\delta \in [0,1]$.

The figure \ref{fig:T3_NPA_P} shows the min-entropies certified with the Bell inequality $T3$ for different additional conditions. In the figure \ref{fig:T3_cert_P} lower-bounds on the certified min-entropy obtained by theorem \ref{DWtoBItheorem} from the NPA hierarchy with additional condition $P(a,b|x,y)=P(\neg a,\neg b|x,y)$ are plotted. These values assume that the untrusted vendor uses the strategy P (see the section \ref{sec:binaryDW}). Figures \ref{fig:T3_NPA_delta} and \ref{fig:T3_cert} contains the relevant data for the mixed strategy.

\subsection{T2}
\label{sub:T2}

A simple Bell inequality is obtained from the symmetric dimension witness of the $2$ to $1$ QRAC used in \cite{DW3,DW4}. It has the following form
\begin{equation}
    \label{T2DWfull}
    \begin{aligned}
        W^{\prime}(1, 1) + W^{\prime}(1, 2) + W^{\prime}(2, 1) - W^{\prime}(2, 2),
    \end{aligned}
\end{equation}
where $W^{\prime}$ is defined in the equation (\ref{CtoWfull}) and $\delta = 1$. The reduced form of this dimension witness is
\begin{equation}
	\label{T2DW}
	\begin{aligned}
		W(1, 1) + W(1, 2) + W(2, 1) - W(2, 2),
	\end{aligned}
\end{equation}
where $W$ is defined by the equation (\ref{CtoW}) and $\delta = 1$. Robustness of the reduced version has been already investigated in \cite{HWL13}, in the figure $4$. The randomness certified by these two dimension witnesses is lower-bounded by the values obtained with the following two Bell inequalities. For the dimension witness defined in the equation (\ref{T2DWfull}), we use a Bell inequality
\begin{equation}
    \label{T2BIfull}
    \begin{aligned}
        & \frac{1}{2} \left( C(1, 1) + C(1, 2) + C(2, 1) - C(2, 2) \right. \\
        & \left. + C(3, 1) + C(3, 2) + C(4, 1) - C(4, 2) \right),
    \end{aligned}
\end{equation}
and for the dimension witness from the equation (\ref{T2DW}),
\begin{equation}
    \label{T2BI}
    \begin{aligned}
        T2 \equiv C(1, 1) + C(1, 2) + C(2, 1) - C(2, 2).
    \end{aligned}
\end{equation}
The operator defined in the equation (\ref{T2BI}) is exactly the CHSH Bell operator. Lower bounds for this case are shown in Figs \ref{fig:T2_NPA_P}, \ref{fig:T2_cert_P}, \ref{fig:T2_NPA_delta} and \ref{fig:T2_cert}. 

The reduced witness (\ref{T2DW}) has recently been experimentally realized \cite{M-expdimwit}. The values obtained in this experiment refer to $p = 0.974$ (5.51 in the scaling used there) and $p = 0.984$ (5.56), concluded therein to certify $0.0595$ and $0.082$ bits of randomness, respectively. If the reduction had not been performed, then only $0.0567$ and $0.0305$ would have been certified.

\subsection{modCHSH}
\label{sub:modCHSH}

In \cite{MP13} the following Bell operator is investigated:
\begin{equation}
    \label{modCHSHBI}
    \begin{aligned}
        modCHSH &\equiv C(1, 2) + C(1, 3) \\
        & + C(2, 1) + C(2, 2) - C(2, 3).
    \end{aligned}
\end{equation}
This Bell operator is similar in the form to the dimension witness introduced in \cite{DW2}. Since the relevant Bell inequality is very robust in certifying the randomness, the dimension witness with randomness lower-bounded by it, may also be expected to be robust. Assuming $P(a|x)=\frac{1}{2}$, we turn it into the following dimension witness
\begin{equation}
	\label{modCHSHDWfull}
	W^{\prime}(1, 2) + W^{\prime}(1, 3) + W^{\prime}(2, 1) + W^{\prime}(2, 2) - W^{\prime}(2, 3).
\end{equation}
Since this dimension witness is symmetric, we follow the steps which lead from the expression (\ref{CtoWfull}), to the expression (\ref{CtoW}), to obtain the following reduced dimension witness
\begin{equation}
	\label{modCHSHDW}
	W(1, 2) + W(1, 3)  + W(2, 1) + W(2, 2) - W(2, 3).
\end{equation}
If we start with the dimension witness defined in the equation (\ref{modCHSHDWfull}), and do not use the symmetry, we get the following lower-bounding Bell inequality
\begin{equation}
    \label{modCHSHBIfull}
    \begin{aligned}
        & \frac{1}{2} \left( C(1, 2) + C(1, 3) + C(2, 1) + C(2, 2) - C(2, 3) \right. \\
        & \left. + C(3, 2) + C(3, 3) + C(4, 1) + C(4, 2) - C(4, 3) \right).
    \end{aligned}
\end{equation}

The dimension witness from the equation (\ref{modCHSHDWfull}) lower-bounds the dimension witness from the equation (\ref{modCHSHDW}), and thus both are lower-bounded (in the sense of the theorem \ref{DWtoBItheorem} and the conjecture below it) by the Bell inequality from the equation (\ref{modCHSHBIfull}), but only the second dimension witness is proved to be lower-bounded by $modCHSH$ (see the equation (\ref{modCHSHBI})). Lower-bounds for this set of DI and SDI protocols are shown in Figs \ref{fig:modCHSH_NPA_P}, \ref{fig:modCHSH_cert_P}, \ref{fig:modCHSH_NPA_delta}, and \ref{fig:modCHSH_cert}.

\section{Conclusions}

In this paper we explained in more details the ideas from our previous paper \cite{HWL13}. In particular all steps of the proof of the theorem \ref{DWtoBItheorem} were provided. A tighter bound, using condition $P(a,b|x,y) = P(\neg a, \neg b|x,y)$ in DI scheme, has been introduced. We have presented a new method of dimension witness reduction and a clear distinction between reduced and full dimension witnesses has been made. Reduced dimension witnesses have been shown to be able to certify more randomness. Min-entropies of several protocols, that had not been considered previously in \cite{HWL13}, were evaluated.

Recently a new method that allows to lower-bound the randomness obtained in a SDI scheme directly, using semi-definite programming, has been introduced in \cite{NTV}. However, the complexity of the algorithm from \cite{NTV} increases significantly with the dimension of Hilbert space while in our case the same computation provides a bound for all dimensions. 

It remains an open question, what are the conditions on a dimension witness under that the adversary has no gain in using the mixed strategy rather than P.

\section{Acknowledgments}
SDP was implemented in OCTAVE using SeDuMi \cite{SeDuMi102} toolbox. This work is supported by IDEAS PLUS (IdP2011 000361), NCN grant 2013/08/M/ST2/00626, FNP TEAM and the National Natural Science Foundation of China (Grant No.11304397). The major part of this work has been written in the forests of Sopot.


\begin{thebibliography}{9}

    \bibitem {Mitnick} K. Mitnick, \textit{The Art of Intrusion}, John Wiley and Sons 0-7645-6959-7 (2005).

    \bibitem {NIST80022} A. Rukhin, J. Soto, J. Nechvatal, M. Smid, E. Barker, S. Leigh, M. Levenson, M. Vangel, D. Banks, A. Heckert, J. Dray, S. Vo, \textit{Special Publication 800-22 Revision 1a}, National Institute of Standards and Technology, U.S. Department of Commerce, available at http://csrc.nist.gov/publications/PubsSPs.html

    \bibitem {OperMinEn} R. Koenig, R. Renner, C. Schaffner,
    IEEE Trans. Inf. Th., vol. \textbf{55}, no. 9 (2009).


    \bibitem {NIST800632} W. E. Burr, D. F. Dodson, E. M. Newton, R. A. Perlner, W. T. Polk, S. Gupta, E. A. Nabbu,
    \textit{NIST Special Publication 800-63-2}, National Institute of Standards and Technology, U.S. Department of Commerce, available at http://csrc.nist.gov/publications/PubsSPs.html
		
		\bibitem {RNGCBT} S. Pironio, A. Acin, S. Massar, A. Boyer de la Giroday, D. N. Matsukevich, P. Maunz, S. Olmschenk, D. Hayes, L. Luo, T. A. Manning, C. Monroe,
    Nature \textbf{464}, 1021 (2010).

    \bibitem {CHSH} J. F. Clauser, M.A. Horne, A. Shimony, R. A. Holt,
    Phys. Rev. Lett. \textbf{23}, 880 (1969).

    \bibitem {ColPHD} R. Colbeck, \textit{Quantum and Relativistic Protocols For Secure Multi-Party Computation.}
    Ph.D. thesis, University of Cambridge, arXiv:0911.3814 (2007).

    \bibitem {CK11} R. Colbeck, A. Kent,
    J. Phys. A: Math. Theor., \textbf{44}(9) 095305 (2011).

    \bibitem {DI} D. Mayers and A. Yao, in FOCS '98: \textit{Proceedings of the 39th Annual Symposium on Foundations of. Computer Science} (IEEE Computer Society, Washington, DC, USA), 503 (1998).
		
		\bibitem {DW1} H.-W. Li, Z.-Q. Yin, Y.-C. Wu, X.-B. Zou, S. Wang, W. Chen, G.-C. Guo, Z.-F. Han,
    Phys. Rev. A \textbf{84}, 034301 (2011).

    \bibitem {SDI} Y.-C. Liang, T. Vertesi, N. Brunner,
    Phys. Rev. A \textbf{83}, 022108 (2011).

    \bibitem {SDI_effects} Y.-K. Wang, S.-J. Qin, T.-T. Song, F.-Z. Guo, W. Huang, H.-J. Zuo,
    Phys. Rev. A \textbf{89}, 032312 (2014).

    \bibitem {IDQ} www.idquantique.com

    \bibitem {QLabs} www.quintessencelabs.com

    \bibitem {DW4} N. Brunner, S. Pironio, A. Acin, N. Gisin, A. A. Methot, V. Scarani,
    Phys. Rev. Lett. \textbf{100}, 210503 (2008).

    \bibitem {DW2} R. Gallego, N. Brunner, C. Hadley, A. Acin,
    Phys. Rev. Lett. \textbf{105}, 230501 (2010).

    \bibitem {DW3} M. Paw\l{}owski, N. Brunner,
    Phys. Rev. A \textbf{84}, 010302(R) (2011).

    \bibitem {DW5} M. Dall'Arno, E. Passaro, R. Gallego, A. Acin,
    Phys. Rev. A \textbf{86}, 042312 (2012).

    \bibitem {Trevisan01} L. Trevisan,
    \textit{Journal of the ACM} \textbf{48}, 860 (2001).

    \bibitem {DPVR09} A. De, C. Portmann, T. Vidick, R. Renner,
    SIAM Journal on Computing \textbf{41}(4), 915 (2012).

    \bibitem {TRSS10} M. Tomamichel, C. Schaffner, A. Smith, R. Renner,
    IEEE Trans. Inf. Theory \textbf{57}(8), (2011).
		
		\bibitem {HWL13} H.-W. Li, P. Mironowicz, M. Paw\l{}owski, Z.-Q. Yin, Y.-C. Wu, S. Wang, W. Chen, H.-G. Hu, G.-C. Guo, Z.-F. Han
    Phys. Rev. A \textbf{87}, 020302(R) (2013).
		
		\bibitem {BC88} S.L. Braunstein, C.M. Caves,
    Phys. Rev. Lett. \textbf{61}, 662 (1988).

    \bibitem {MP13} P. Mironowicz, M. Paw\l{}owski,
    Phys. Rev. A \textbf{88}, 032319 (2013).

    \bibitem {NPA07} M. Navascues, S. Pironio, A. Acin,
    Phys. Rev. Lett. \textbf{98}, 010401 (2007).

    \bibitem {NPA08} M. Navascues, S. Pironio, A. Acin,
    New J. Phys. \textbf{10}, 073013 (2008).

    \bibitem {QRAC} A. Ambainis, A. Nayak, A. Ta-shma, and U. Vazirani,
    \textit{Dense quantum coding and a lower bound for 1-way quantum automata}, in Proceedings of 31st ACM Symposium on Theory of Computing, 376 (1999).

    \bibitem {HWL12} H.-W. Li, M. Paw\l{}owski, Z.-Q. Yin, G.-C. Guo, Z.-F. Han,
    Phys. Rev. A \textbf{85}, 052308 (2012).

    \bibitem {CGLMP} D. Collins, N. Gisin, N. Linden, S. Massar, S. Popescu,
    Phys. Rev. Lett. \textbf{88}, 040404 (2002).

    \bibitem {SeDuMi102} J.F. Sturm,
    Optimization Methods and Software \textbf{11}, 625 (1999).

    \bibitem {IntPoint} J.F. Sturm,
    Optimization Methods and Software \textbf{17}, 6 (2002).
	
    \bibitem{M-expdimwit} J. Ahrens, P. Badziag, M. Pawlowski, M. Zukowski, M. Bourennane, Phys. Rev. Lett. {\bf 112}, 140401 (2014).

		\bibitem {NTV} M. Navascues, G. de la Torre, T. Vertesi,
    arXiv:1308.3410

\end{thebibliography}
\end{document}